\documentclass[12pt]{article}
\usepackage[margin= 2 cm, bottom=20mm,footskip=10mm,top=20mm]{geometry}
\usepackage{hyperref}
\usepackage{amsthm}
\usepackage{amsmath}
\usepackage{amssymb}
\usepackage{amsfonts}
\usepackage{cite}
\usepackage{epsfig}
\usepackage{mathtools}
\usepackage{enumitem}
\usepackage{stackengine,graphicx}

\newtheorem{theorem}{Theorem}
\newtheorem{lemma}[theorem]{Lemma}

\newtheorem{corollary}[theorem]{Corollary}
\newtheorem{problem}{Problem}

\newcommand{\RR}{\mathbb{R}}
\newcommand{\FF}{\mathbb{F}}
\newcommand{\II}{\mathcal{I}}
\newcommand{\NN}{\mathbb{N}}

\newcommand{\calO}{\mathcal{O}}
\newcommand{\ZZ}{\mathbb{Z}}
\newcommand{\QQ}{\mathbb{Q}}

\DeclareMathOperator{\rank}{rank}
\DeclareMathOperator{\cl}{cl}
\DeclareMathOperator{\dd}{dd}

\DeclareMathOperator{\td}{td}

\DeclareMathOperator{\ec}{ec}
\DeclareMathOperator{\cd}{cd}
\DeclareMathOperator{\cdd}{cdd}
\DeclareMathOperator{\csd}{c^{*}\hspace{-3pt}d}
\DeclareMathOperator{\csdd}{c^{*}\hspace{-3pt}dd}

\def\ve#1{\mathchoice{\mbox{\boldmath$\displaystyle\bf#1$}}
{\mbox{\boldmath$\textstyle\bf#1$}}
{\mbox{\boldmath$\scriptstyle\bf#1$}}
{\mbox{\boldmath$\scriptscriptstyle\bf#1$}}}

\newcommand\vex{{\ve x}}
\newcommand\vey{{\ve y}}

\def\G{\mathcal{G}}
\def\C{\mathcal{C}}

\def\Ker{\textrm{Ker}}

\begin{document}
\title{Characterization of matrices with bounded Graver bases and depth parameters and
       applications to integer programming
       \thanks{
       The first author was partially supported by the Polish National Science Center grant (BEETHOVEN; UMO-2018/31/G/ST1/03718).
       The third author was supported by the European Research Council (ERC) under the European Union's Horizon 2020 research and innovation programme (grant agreement No 648509).
       This publication reflects only its authors' view; the ERC Executive Agency is not responsible for any use that may be made of the information it contains.
       The third and fourth authors were supported by the MUNI Award in Science and Humanities (MUNI/I/1677/2018) of the Grant Agency of Masaryk University.}
       \thanks{
       The extended abstract containing the results presented in this paper has appeared in the proceedings of the 49th International Colloquium Automata, Languages and Programming (ICALP'22).
       }
       }
\author{Marcin Bria\'nski\thanks{Theoretical Computer Science Department, Faculty of Mathematics and Computer Science, Jagiellonian University, Krak\'ow, Poland, E-mail: \texttt{marcin.brianski@doctoral.uj.edu.pl}}\and
        Martin Kouteck\'y\thanks{Computer Science Institute, Charles University, Prague, Czech Republic. E-mail: {\tt koutecky@iuuk.mff.cuni.cz}.}\and
        Daniel Kr\'al'\thanks{Faculty of Informatics, Masaryk University, Brno, Czech Republic. E-mails: {\tt dkral@fi.muni.cz} and {\tt kristyna.pekarkova@mail.muni.cz}.}\and
\newcounter{kth}
\setcounter{kth}{5}
	Krist\'yna Pek\'arkov\'a$^\fnsymbol{kth}$\and
	Felix Schr\"oder\thanks{Institute of Mathematics, Technical University Berlin, Berlin, Germany. E-Mail: {\tt fschroed@math.tu-berlin.de}}}
        
\date{}

\maketitle

\begin{abstract}
An intensive line of research on fixed parameter tractability of integer programming
is focused on exploiting the relation between the sparsity of a constraint matrix $A$ and
the norm of the elements of its Graver basis.
In particular, integer programming is fixed parameter tractable
when parameterized by the primal tree-depth and the entry complexity of $A$, and
when parameterized by the dual tree-depth and the entry complexity of $A$;
both these parameterization imply that $A$ is sparse, in particular,
the number of its non-zero entries is linear in the number of columns or rows, respectively.

We study preconditioners transforming a given matrix to a row-equivalent sparse matrix if it exists and
provide structural results characterizing the existence of a sparse row-equivalent matrix in terms of
the structural properties of the associated column matroid.
In particular, our results imply that
the $\ell_1$-norm of the Graver basis is bounded by a function of the maximum $\ell_1$-norm of a circuit of $A$.
We use our results to design a parameterized algorithm that constructs a matrix row-equivalent to an input matrix $A$ that
has small primal/dual tree-depth and entry complexity if such a row-equivalent matrix exists.

Our results yield parameterized algorithms for integer programming
when parameterized by the $\ell_1$-norm of the Graver basis of the constraint matrix,
when parameterized by the $\ell_1$-norm of the circuits of the constraint matrix,
when parameterized by the smallest primal tree-depth and entry complexity of a matrix row-equivalent to the constraint matrix, and
when parameterized by the smallest dual tree-depth and entry complexity of a matrix row-equivalent to the constraint matrix.
\end{abstract}

\section{Introduction}

Integer programming is a problem of fundamental importance in combinatorial optimization
with many theoretical and practical applications.
From the computational complexity point of view, integer programming is very hard:
it is one of the 21 problems shown to be NP-complete in the original paper on NP-completeness by Karp~\cite{Kar72} and
remains NP-complete even when the entries of the constraint matrix are zero and one only.
On the positive side,
Kannan and Lenstra~\cite{Kan87,Lens83} showed that integer programming is polynomially solvable in fixed dimension,
i.e., with a fixed number of variables.
Another prominent tractable case is when the constraint matrix is totally unimodular,
i.e., all determinants of its submatrices are equal to $0$ or $\pm 1$,
in which case all vertices of the feasible region are integral and
so linear programming algorithms can be applied.

Integer programming (IP) is known to be tractable for instances
where the constraint matrix of an input instance enjoys a certain block structure.
The two most important cases are the cases of $2$-stage IPs due to Hemmecke and Schultz~\cite{HemS03},
further investigated in particular in~\cite{AscH07,Kle21,CslEPVW21,JanKL21,KleR21,KouLO18}, and
$n$-fold IPs introduced by De Loera et al.~\cite{DHOW} and 
further investigated in particular in~\cite{HemOR13,CheM18,EisHK18,JanLR20,CslEHRW21,KouLO18}.
IPs of this kind appear in various contexts, see e.g.~\cite{JanKMR21,KnoK18,KnoK20,SchSV98}.
These (theoretical) tractability results complement well a vast number of empirical results
demonstrating tractability of instances with a block structure,
e.g.~\cite{BorFM98,BerCCFLMT15,KhaEE18,FerH98,AykPC04,WeiK71,WanR13,GamL10,VanW10}.

There tractability results on IPs with sparse constraint matrices
can be unified and generalized using depth and width parameters of graphs derived from constraint matrices.
Ganian and Ordyniak~\cite{GanO18} initiated this line of study by showing that
IPs with bounded primal tree-depth $\td_P(A)$ of a constraint matrix $A$ and
bounded coefficients of the constraint matrix $A$ and the right hand side $b$ can be solved efficiently.
Levin, Onn and the second author~\cite{KouLO18}
widely generalized this result by showing that
IPs with bounded $\|A\|_\infty$ and bounded primal tree-depth $\td_P(A)$ or
dual tree-depth $\td_D(A)$ of the constraint matrix $A$ can be solved efficiently;
such IPs include $2$-stage IPs, $n$-fold IPs, and their generalizations.

Most of the existing algorithms for IPs assume that
the input matrix is already given in its sparse form.
This is a substantial drawback as existing algorithms cannot be applied to instances that
are not sparse but can be transformed to an equivalent sparse instance.
For example, the matrix in the left below, whose dual tree-depth is $5$,
can be transformed by elementary row operations to the matrix with dual tree-depth $2$ given in the right;
a formal definition of tree-depth is given in Subsection~\ref{subsec:graph}, however, 
just the visual appearance of the two matrices indicates which is likely to be more amenable to algorithmic techniques.
\[\left(\begin{matrix}
  2 & 2 & 1 & 2 & 1 & 3 & 1 \\
  2 & 1 & 1 & 1 & 2 & 1 & 1 \\
  2 & 2 & 2 & 2 & 2 & 2 & 1 \\
  2 & 1 & 1 & 2 & 2 & 1 & 1 \\
  2 & 2 & 1 & 2 & 1 & 3 & 2
  \end{matrix}\right)
  \quad
  \rightarrow
  \quad
  \left(\begin{matrix}
  2 & 1 & 0 & 1 & 1 & 2 & 1 \\
  0 & 1 & 1 & 0 & 0 & 1 & 0 \\
  1 & 0 & 0 & 0 & 0 & 0 & 0 \\
  0 & 0 & 0 & 1 & 0 & 0 & 0 \\
  0 & 0 & 0 & 0 & 2 & 0 & 1
  \end{matrix}\right)
  \]
This transformation is an example of a preconditioner that
transforms an instance of integer programming to an equivalent one that
is more amenable to existing methods for solving integer programming and
whose existence we investigate in this paper.

Preconditioning a problem to make it computationally simpler
is a ubiquitous preprocessing step in mathematical programming solvers.
An interesting link between matroid theory and preconditioners to sparsity of matrices
was exhibited by Chan and Cooper together with the second, third and fourth authors~\cite{ChaCKKP19,ChaCKKP20}.
In particular, they proved the following structural characterization of matrices that are \emph{row-equivalent},
i.e., can be transformed by elementary row operations, to a matrix with small dual tree-depth:
a matrix is row-equivalent to one with small dual tree-depth if and only if
the column matroid of the matrix has small contraction$^*$-depth (see Theorem~\ref{thm:orig-eq} below).
In this paper,
we further explore this uncharted territory
by providing a structural characterization of matrices row-equivalent to matrices with small primal tree-depth,
designing efficient algorithms for finding preconditioners with respect to both primal and dual tree-depth, and
relating complexity of circuits and Graver basis of constraint matrices.

\subsection{Our contribution}
\label{subsec:contrib}

We now describe the results presented in this paper in detail.
We opted not to interrupt the presentation of our results with various notions, some of which may be standard for some readers, and
rather collect all definitions in a single section---Section~\ref{sec:prelim}.
We remark that the primal tree-depth of a matrix $A$
is a structural parameter that measures the complexity of interaction between the columns of $A$, and
the dual tree-depth of a matrix $A$ measures the complexity of interaction between the rows of $A$.

\subsubsection{Characterization of depth parameters}

Observe that the column matroid of the matrix is preserved by row operations,
i.e., the column matroid of row-equivalent matrices is the same.
The main structural result of~\cite{ChaCKKP19,ChaCKKP20}
is the following characterization of the existence of a row-equivalent matrix with small dual tree-depth
in terms of the structural parameter of the column matroid~\cite[Theorem 1]{ChaCKKP20}.
We remark that the term branch-depth was used in~\cite{ChaCKKP19,ChaCKKP20} in line with the terminology from~\cite{KarKLM17}
but as there is a competing notion of branch-depth~\cite{DevKO20},
we decided to use a different name for this depth parameter throughout the paper to avoid confusion.

\begin{theorem}
\label{thm:orig-eq}
For every non-zero matrix $A$, 
it holds that
the smallest dual tree-depth of a matrix row-equivalent to $A$ is equal to the contraction$^*$-depth of $M(A)$,
i.e., $\td_D^*(A)=\csd(A)$.
\end{theorem}

We discover structural characterizations of the existence of a row-equivalent matrix with small primal tree-depth and
the existence of a row-equivalent matrix with small incidence tree-depth.

\begin{theorem}
\label{thm:tdP}
For every matrix $A$,
it holds that
the smallest primal tree-depth of a matrix row-equivalent to $A$ is equal to the deletion-depth of $M(A)$,
i.e., $\td_P^*(A)=\dd(A)$.
\end{theorem}

\begin{theorem}
\label{thm:tdI}
For every matrix $A$,
it holds that
the smallest incidence tree-depth of a matrix row-equivalent to $A$ is equal to contraction$^*$-deletion-depth of $M(A)$ increased by one,
i.e., $\td_I^*(A)=\csdd(A)+1$.
\end{theorem}

\subsubsection{Interplay of circuit and Graver basis complexity}

Graver bases play an essential role in designing efficient algorithms for integer programming.
We show that the maximum $\ell_1$-norm of a circuit of a matrix $A$ and
the maximum $\ell_1$-norm of an element of the Graver basis of $A$,
which are denoted by $c_1(A)$ and $g_1(A)$, respectively, are functionally equivalent.

\begin{theorem}
\label{thm:equiv1}
There exists a function $f_1:\NN\to\NN$ such that the following holds for every rational matrix $A$ with $\dim\ker A>0$:
\[ c_1(A)\le g_1(A)\le f_1(c_1(A)). \]
\end{theorem}

The parameter $c_1(A)$ can be related to dual tree-depth and entry complexity as follows (we have opted
throughout the paper to use entry complexity rather than $\|A\|_\infty$ as this permits to formulate
our results for rational matrices rather than integral matrices, which is occasionally more convenient).

\begin{theorem}
\label{thm:equiv2}
Every rational matrix $A$ with $\dim\ker A>0$
is row-equivalent to a rational matrix $A'$ with $\td_D(A')\le c_1(A)^2$ and $\ec(A')\le 2\lceil \log_2 (c_1(A)+1)\rceil$.
\end{theorem}

Our results together with Theorem~\ref{thm:boundg} imply that
the following statements are equivalent for every rational matrix $A$:
\begin{itemize}
\item The $\ell_1$-norm of every circuit of $A$, i.e., $c_1(A)$, is bounded.
\item The $\ell_1$-norm of every element of the Graver basis of $A$, i.e., $g_1(A)$, is bounded.
\item The matrix $A$ is row-equivalent to a matrix with bounded dual tree-depth and bounded entry complexity.
\item The contraction$^*$-depth of the matroid $M(A)$ is bounded, and
      the matrix $A$ is row-equivalent to a matrix with bounded entry complexity (with any dual tree-depth).
\end{itemize}

\subsubsection{Algorithms to compute matrices with small depth parameters}

We also construct parameterized algorithms for transforming an input matrix to a row-equivalent matrix
with small tree-depth and entry complexity if one exists. 
First, we design a parameterized algorithm for computing
a row-equivalent matrix with small primal tree-depth and small entry complexity if one exists.

\begin{theorem}
\label{thm:alg-tdP}
There exists a function $f:\NN^2\to\NN$ and a fixed parameter algorithm for the parameterization by $d$ and $e$ that
for a given rational matrix $A$:
\begin{itemize}
\item either outputs that $A$ is not row-equivalent to a matrix with primal tree-depth at most $d$ and entry complexity at most $e$, or
\item outputs a matrix $A'$ that is row-equivalent to $A$, its primal tree-depth is at most $d$ and entry complexity is at most $f(d,e)$.
\end{itemize}
\end{theorem}

The following algorithm for computing
a row-equivalent matrix with small dual tree-depth was presented in~\cite{ChaCKKP19,ChaCKKP20}.

\begin{theorem}
\label{thm:orig-alg}
There exists a function $f:\NN^2\to\NN$ and a fixed parameter algorithm for the parameterization by $d$ and $e$ that
for a given rational matrix $A$ with entry complexity at most $e$:
\begin{itemize}
\item either outputs that $A$ is not row-equivalent to a matrix with dual tree-depth at most $d$, or
\item outputs a matrix $A'$ that is row-equivalent to $A$, its dual tree-depth is at most $d$ and entry complexity is at most $f(d,e)$.
\end{itemize}
\end{theorem}

We improve the algorithm by replacing the parameterization by the entry complexity of an input matrix
with the parameterization by the entry complexity of the to be constructed matrix.
Note that if a matrix $A$ has entry complexity $e$ and is row-equivalent to a matrix with dual tree-depth $d$,
then Theorem~\ref{thm:orig-alg} yields that
$A$ is row-equivalent to a matrix with dual tree-depth $d$ and entry complexity bounded by a function of $d$ and $e$.
Hence, the algorithm given below applies to a wider set of input matrices than the algorithm from Theorem~\ref{thm:orig-alg}.

\begin{theorem}
\label{thm:alg-tdD}
There exists a function $f:\NN^2\to\NN$ and a fixed parameter algorithm for the parameterization by $d$ and $e$ that,
for a given rational matrix $A$:
\begin{itemize}
\item either outputs that $A$ is not row-equivalent to a matrix with dual tree-depth at most $d$ and entry complexity at most $e$, or
\item outputs a matrix $A'$ that is row-equivalent to $A$, its dual tree-depth is at most $d$ and entry complexity is at most $f(d,e)$.
\end{itemize}
\end{theorem}

We point out the following difference between the cases of primal and dual tree-depth.
As mentioned, if a matrix $A$ has entry complexity $e$ and is row-equivalent to a matrix with dual tree-depth $d$,
then $A$ is row-equivalent to a matrix with dual tree-depth $d$ and entry complexity bounded by a function of $d$ and $e$.
However, the same is not true in the case of primal tree-depth.
The entry complexity of every matrix with primal tree-depth equal to one that is row-equivalent to the following matrix $A$
is linear in the number of rows of $A$,
quite in a contrast to the case of dual tree-depth.
\[\left(\begin{matrix}
  1 & 2 & 0 & 0 & \cdots & 0 & 0 & 0 & 0 \\
  0 & 1 & 2 & 0 & \cdots & 0 & 0 & 0 & 0 \\
  0 & 0 & 1 & 2 & \cdots & 0 & 0 & 0 & 0 \\
  \vdots & \vdots & & \ddots & \ddots & & & \vdots & \vdots \\
  \vdots & \vdots & & & \ddots & \ddots & & \vdots & \vdots \\
  0 & 0 & 0 & 0 & \cdots & 1 & 2 & 0 & 0 \\
  0 & 0 & 0 & 0 & \cdots & 0 & 1 & 2 & 0 \\
  0 & 0 & 0 & 0 & \cdots & 0 & 0 & 1 & 2
  \end{matrix}\right)
  \]

\subsubsection{Fixed parameter algorithms for integer programming}

One of the open problems in the area, e.g.~discussed during the Dagstuhl workshop 19041 ``New Horizons in Parameterized Complexity'',
has been whether integer programming is fixed parameter tractable when parameterized by $g_1(A)$,
i.e., by the $\ell_1$-norm of an element of the Graver basis of the constraint matrix $A$.
Our results on the interplay of dual tree-depth, the circuit complexity and the Graver basis complexity of a matrix
yield an affirmative answer.
The existence of appropriate preconditioners that we establish in this paper
implies that integer programming is fixed parameter tractable when parameterized by
\begin{itemize}
\item $g_1(A)$, i.e., the $\ell_1$-norm of the Graver basis of the constraint matrix,
\item $c_1(A)$, i.e., the $\ell_1$-norm of the circuits of the constraint matrix,
\item $\td^*_P(A)$ and $\ec(A)$, i.e., the smallest primal tree-depth and entry complexity of a matrix row-equivalent to the constraint matrix, and
\item $\td^*_D(A)$ and $\ec(A)$, i.e., the smallest dual tree-depth and entry complexity of a matrix row-equivalent to the constraint matrix.
\end{itemize}
We believe that our new tractability results significantly enhance the toolbox of tractable IPs as
the nature of our tractability conditions substantially differ from prevalent block-structured sparsity-based tractability conditions.
The importance of availability of various forms of tractable IPs can be witnessed by $n$-fold IPs,
which were shown fixed-parameter tractable in~\cite{HemOR13}, and,
about a decade later,
their applications has become ubiquitous,
see e.g.~\cite{BreKKN19,BreFKKN21,CheM18,CheCZ21,HerMNS21,JanLM20,JanKMR21,KnoK18, KnoKM17_2}.

\subsubsection{Hardness results}

As our algorithmic results involve computing depth decompositions of matroids for various depth parameters in a parameterized way,
we establish computational hardness of these parameters in Theorem~\ref{thm:NPc}, primarily for the sake of completeness of our exposition.
In particular, computing the following matroid parameters is NP-complete:
\begin{itemize}
\item deletion-depth,
\item contraction-depth,
\item contraction-deletion-depth,
\item contraction$^*$-depth, and
\item contraction$^*$-deletion-depth.
\end{itemize}

\section{Preliminaries}
\label{sec:prelim}

In this section, we fix the notation used throughout the paper.
We start with general notation and we then fix the notation related to graphs, matrices and matroids.

The set of all positive integers is denoted by $\NN$ and
the set of the first $k$ positive integers by $[k]$.
If $A$ is a linear space, we write $\dim A$ for its dimension.
If $K$ is a subspace of $A$,
the \emph{quotient space} $A/K$ is the linear space of the dimension $\dim A-\dim K$ that
consists of cosets of $A$ given by $K$ with the natural operations of addition and scalar multiplication;
see e.g.~\cite{Hal93} for further details.
The quotient space $A/K$ can be associated with a linear subspace of $A$ of dimension $\dim A-\dim K$
formed by exactly a single vector from each coset of $A$ given by $K$;
we will often view the quotient space as such a subspace of $A$ and write $w+K$ for the coset containing a vector $w$.
For example, if $A$ is $\RR^3$ and $K$ is the linear space generated by $(0,0,1)$,
$A/K$ can be associated with (or viewed as) the $2$-dimensional space formed by vectors $(x,y,0)$, $x,y\in\RR$.

\subsection{Graphs}
\label{subsec:graph}

All graphs considered in this paper are loopless simple graphs unless stated otherwise.
If $G$ is a graph, then we write $V(G)$ and $E(G)$ for the vertex set and the edge set of $G$, respectively.
If $W$ is a subset of vertices of a graph $G$,
then $G\setminus W$ is the graph obtained by removing the vertices of $W$ (and all edges incident with them), and
$G[W]$ is the graph obtained by removing all vertices not contained in $W$ (and all edges incident with them).
If $F$ is a subset of edges of a graph $G$,
then $G\setminus F$ is the graph obtained by removing the edges contained in $F$ and
$G/F$ is the graph obtained by contracting all edges contained in $F$ and removing resulting loops and parallel edges (while
keeping one edge from each group of parallel edges).

We next define the graph parameter \emph{tree-depth}, which is the central graph parameter in this paper.
The \emph{height} of a rooted tree is the maximum number of vertices on a path from the root to a leaf, and
the \emph{height} of a rooted forest, i.e., a graph whose each component is a rooted tree,
is the maximum height of its components.
The \emph{depth} of a rooted tree is the maximum number of edges on a path from the root to a leaf, and
the \emph{depth} of a rooted forest is the maximum depth of its components.
Note that the height and the depth of a rooted tree always differ by one;
we use both notions to avoid cumbersome way of expressing that would otherwise require adding or subtracting one.
The \emph{closure~$\cl(F)$} of a rooted forest $F$ is the graph obtained by adding edges from each vertex to all its descendants.
Finally,
the \emph{tree-depth~$\td(G)$} of a graph~$G$ is the minimum height of a rooted forest~$F$ such that
the closure~$\cl(F)$ of the rooted forest~$F$ contains~$G$ as a subgraph.
See Figure~\ref{fig:tdepth} for an example.
It can be shown that the path-width, and so the tree-width, of any graph is at most its tree-depth decreased by one;
see e.g.~\cite{CygFKLMPPS15} for a more detailed discussion of the relation of tree-depth, path-width and tree-width, and
their algorithmic applications.

\begin{figure}
\begin{center}
\epsfbox{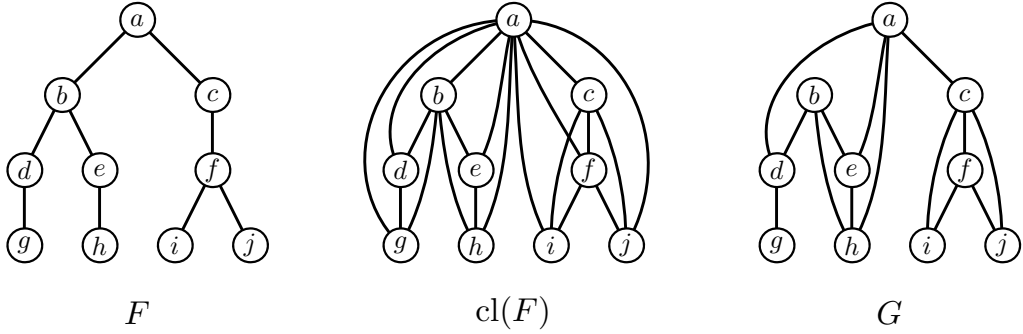}
\end{center}
\caption{A rooted forest $F$ consisting of a single tree and its closure $\cl(F)$,
         which shows that tree-depth of the depicted graph $G$ is (at most) four.}
\label{fig:tdepth}
\end{figure}

\subsection{Matroids}

We next review basic definitions from matroid theory;
we refer to the book of Oxley~\cite{Oxl11} for detailed exposition.
A \emph{hereditary} collection of subsets of a set is a collection closed under taking subsets;
in particular, every non-empty hereditary collection of subsets contains the empty set.
A \emph{matroid} $M$ is a pair~$(X,\II)$,
where~$\II$ is a non-empty hereditary collection of subsets of~$X$ that satisfies the \emph{augmentation axiom},
i.e., if $X'\in\II$, $X''\in\II$ and $|X'|<|X''|$,
then there exists an element~$x\in X''\setminus X'$ such that $X'\cup\{x\}\in\II$.
The set $X$ is the \emph{ground set} of $M$ and the sets contained in~$\II$ are referred to as \emph{independent}.
we often refer to elements of the ground set of $M$ to as elements of the matroid $M$, and
if $e$ is an element of (the ground set of) $M$, we also write $e\in M$.
Two important examples of matroids are vector matroids and graphic matroids.
A \emph{vector matroid} is a matroid whose ground set is formed by vectors and
independent sets are precisely sets of linearly independent vectors (note that the augmentation axiom follows
from the Steinitz exchange lemma).
A \emph{graphic matroid} is a matroid whose ground set is formed by edges of a graph and
independent sets are precisely acyclic sets of edges, i.e., sets not containing a cycle.

The \emph{rank} of a subset $X'$ of the ground set $X$,
which is denoted by $r_M(X')$ or simply by $r(X')$ if $M$ is clear from the context,
is the maximum size of an independent subset of~$X'$ (it can be shown that
all maximal independent subsets of $X'$ have the same cardinality);
the \emph{rank} of the matroid $M$, which is denoted by $r(M)$, is the rank of its ground set.
Note that in the case of vector matroids,
the rank of $X'$ is exactly the dimension of linear space generated by $X'$.
A \emph{basis} of a matroid $M$ is a maximal independent subset of the ground set of $M$ and
a \emph{circuit} is a minimal subset of the ground set of $M$ that is not independent.
In particular, if $X'$ is a circuit of $M$, then $r(X')=|X'|-1$ and every proper subset of $X'$ is independent.
An element $x$ of a matroid $M$ is a \emph{loop} if $r(\{x\})=0$, 
an element $x$ is a \emph{bridge} if it is contained in every basis of $M$, and
two elements $x$ and $x'$ are \emph{parallel} if~$r(\{x\})=r(\{x'\})=r(\{x,x'\})=1$.
Note that in the case of vector matroids,
two non-loop elements are parallel if and only if they are non-zero multiple of each other.
If $M$ is a matroid with ground set $X$,
the \emph{dual matroid}, which is denoted by $M^*$ is the matroid with the same ground set $X$ such that
$X'\subseteq X$ is independent in $M^*$ if and only if $r_M(X\setminus X')=r(M)$;
in particular, $r_{M^*}(X')=r_M(X\setminus X')+|X'|-r(M)$ for every $X'\subseteq X$.

For a field $\FF$, we say that a matroid $M$ is \emph{$\FF$-representable}
if every element of $M$ can be assigned a vector from $\FF^{r(M)}$ in such a way that
a subset of the ground set of $M$ is independent if and only if the set of assigned vectors is linearly independent.
In particular, an element of $M$ is a loop if and only if it is assigned the zero vector and
two non-loop elements of $M$ are parallel if and only if they are non-zero multiples of each other.
Such an assignment of vectors of $\FF^{r(M)}$ to the elements of $M$ is an \emph{$\FF$-representation} of $M$.
Clearly, a matroid $M$ is $\FF$-representable if and only if
it is isomorphic to the vector matroid given by its $\FF$-representation.
Matroids representable over the $2$-element field are referred to as \emph{binary} matroids.
We say that a matroid $M$ is \emph{$\FF$-represented}
if the matroid $M$ is given by its $\FF$-representation.
If a particular field $\FF$ is not relevant in the context,
we just say that a matroid $M$ is \emph{represented} to express that it is given by its representation.

Let $M$ be a matroid with a ground set $X$.
The matroid $kM$ for $k\in\NN$ is the matroid obtained from $M$ 
by introducing $k-1$ parallel elements to each non-loop element and $k-1$ additional loops for each loop;
informally speaking, every element of $M$ is ``cloned'' to $k$ copies.
Note that a subset $X'$ of the elements of $kM$ is independent if and only if
it does not contains two clones of the same element and
the set of the elements of $M$ corresponding to those contained in $X'$ is independent.
Observe that if $M$ is a vector matroid,
then $kM$ is the vector matroid obtained by adding $k-1$ copies of each vector forming $M$.
Similarly, if $M$ is a graphic matroid associated with a graph $G$,
then $kM$ is the graphic matroid obtained from the graph $G$ by duplicating each edge $k-1$ times.

If $X'\subseteq X$, then the \emph{restriction} of $M$ to $X'$, which is denoted by $M\left[X'\right]$,
is the matroid with the ground set $X'$ such that a subset of $X'$
is independent in $M\left[X'\right]$ if and only if it is independent in $M$.
In particular, the rank of $M\left[X'\right]$ is $r_M(X')$.
For example, if $M$ is a graphic matroid associated with a graph $G$,
then the restriction of $M$ to $X'$ is the graphic matroid associated with the spanning subgraph of $G$ with edge set $X'$.
The matroid obtained from $M$ by \emph{deleting} $X'$ is the restriction of $M$ to $X\setminus X'$ and is denoted by $M\setminus X'$.

The \emph{contraction} of $M$ by $X'$, which is denoted by $M/X'$,
is the matroid with the ground set $X\setminus X'$ such that
a subset $X''$ of $X\setminus X'$ is independent in~$M/X'$ if and only if~$r_M(X''\cup X')=\lvert X''\rvert+r_M(X')$.
If $X'$ is a single element set and $e$ is its only element,
we write $M\setminus e$ and $M/e$ instead of $M\setminus\{e\}$ and $M/\{e\}$, respectively.
If $M$ is a graphic matroid associated with a graph $G$ and $e$ is an edge of $G$,
then $M/e$ is the graphic matroid associated with the graph obtained from $G$ by contracting the edge $e$ (while keeping all resulting loops and parallel edges).
If an $\FF$-representation of $M$ is given and $X'$ is a subset of the ground set of $M$,
then an $\FF$-representation of $M/X'$ can be obtained from the $\FF$-representation of $M$
by considering the representation in the quotient space by the linear hull of the vectors representing the elements of $X'$.
This leads us to the following definition:
if $M$ is an $\FF$-represented matroid and $A$ is a linear subspace of $\FF^{r(M)}$,
then the matroid $M/A$ is the $\FF$-represented matroid with the representation of $M$ in the quotient space by $A$.
Note that the ground sets of $M$ and $M/A$ are the same,
in particular, $M$ and $M/A$ have the same number of elements.

A matroid~$M$ is \emph{connected} if every two distinct elements of~$M$ are contained in a common circuit.
We remark that the property of being contained in a common circuit is transitive~\cite[Proposition 4.1.2]{Oxl11},
i.e., if the pair elements $e$ and $e'$ are contained in a common circuit and
the pair $e'$ and $e''$ is also contained in a common circuit,
then the pair $e$ and $e''$ is also contained in a common circuit.
If $M$ is an $\FF$-represented matroid with at least two elements,
then $M$ is connected if and only if $M$ has no loops and
there do not exist two non-trivial linear spaces~$A$ and $B$ of $\FF^{r(M)}$ such that
$A\cap B$ contains the zero vector only and every element of $M$ is contained in $A$ or $B$ (the linear space $\FF^{r(M)}$
would be the direct sum of the linear spaces $A$ and $B$).
Also observe that if $M$ is a graphic matroid associated with a graph $G$,
then the matroid $M$ is connected if and only if the graph $G$ is $2$-connected.

A \emph{component} of a matroid $M$ is an inclusion-wise maximal connected restriction of $M$;
a component is \emph{trivial} if it consists of a single loop, and it is \emph{non-trivial} otherwise.
If $M$ is a vector matroid, then each non-trivial component of $M$ can be associated with a linear space such that
each element of $M$ is contained in one of the linear spaces and
the linear hull of all elements of $M$ is the direct sum of the linear spaces.
We often identify components of a matroid $M$ with their element sets.
Using this identification,
it holds that a subset $X'$ of a ground set of a matroid $M$ is a component of $M$
if and only if $X'$ is a component of $M^*$ (we use this equivalence
to prove some of our hardness results in Section~\ref{sec:hardness}).
We remark that $(M^*)^*=M$ for every matroid $M$, and
if $e$ is an element of a matroid $M$,
then $(M/e)^*=M^*\setminus e$ and $(M\setminus e)^*=M^*/e$.

\subsection{Matrices}

In this section, we define notation related to matrices.
If $\FF$ is a field, we write $\FF^{m\times n}$ for the set of matrices with $m$ rows and $n$ columns over the field $\FF$.
If $A$ is a rational matrix,
the entry complexity $\ec(A)$ is the maximum length of a binary encoding of its entries,
i.e., the maximum of $\left\lceil\log_2\left(|p|+1\right)\right\rceil+\left\lceil\left(\log_2 |q|+1\right)\right\rceil$
taken over all entries $p/q$ of $A$ (where $p$ and $q$ are always assumed to be coprime).
If $A$ is an integral matrix, then $\ec(A)=\Theta\left(\log\|A\|_{\infty}\right)$.
Throughout the paper, we use the entry complexity rather than the $\ell_\infty$-norm of matrices as
this permits formulating our results for rational matrices rather than integral matrices only.

A rational matrix $A$ is \emph{$z$-integral} for $z\in\QQ$
if every entry of $A$ is an integral multiple of $z$.
We say that two matrices $A$ and $A'$ are \emph{row-equivalent}
if one can be obtained from another by \emph{elementary row operations},
i.e., by repeatedly adding a multiple of one row to another and multiplying a row by a non-zero element.
Observe that if $A$ and $A'$ are row-equivalent matrices, then their kernels are the same.
For a matrix $A$, we define $M(A)$ to be the represented matroid whose elements are the columns of $A$.
Again, if matrices $A$ and $A'$ are row-equivalent, then the matroids $M(A)$ and $M(A')$ are the same.

\begin{figure}
\begin{center}
\epsfbox{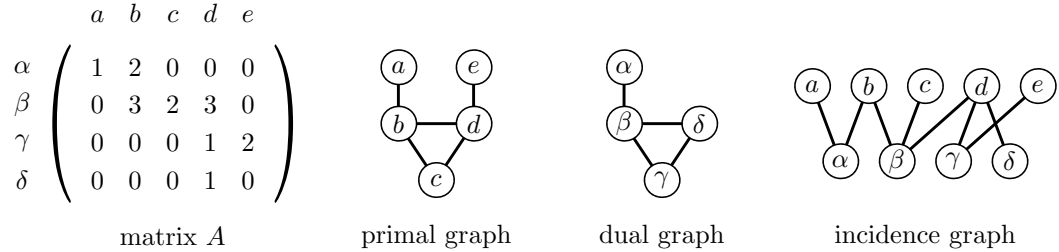}
\end{center}
\caption{The primal graph, the dual graph and the incidence graph of the depicted matrix $A$.}
\label{fig:matrix}
\end{figure}

If $A$ is a matrix, the \emph{primal graph} of $A$ is the graph whose vertices are columns of $A$ and
two vertices are adjacent if there exists a row having non-zero elements in the two columns associated with the vertices;
the \emph{dual graph} of $A$ is the graph whose vertices are rows of $A$ and
two vertices are adjacent if there exists a column having non-zero elements in the two associated rows;
the \emph{incidence graph} of $A$ is the bipartite graph with one part formed by rows of $A$ and the other part by columns of $A$ and
two vertices are adjacent if the entry in the associated row and in the associated column is non-zero.
See Figure~\ref{fig:matrix} for an example.
The \emph{primal tree-depth} of $A$, denoted by $\td_P(A)$, is the tree-depth of the primal graph of $A$,
the \emph{dual tree-depth} of $A$, denoted by $\td_D(A)$, is the tree-depth of the dual graph of $A$, and
the \emph{incidence tree-depth} of $A$, denoted by $\td_I(A)$, is the tree-depth of the incidence graph of $A$.
Finally, $\td_P^*(A)$ is the smallest primal tree-depth of a matrix row-equivalent to $A$,
$\td_D^*(A)$ is the smallest dual tree-depth of a matrix row-equivalent to $A$, and
$\td_I^*(A)$ is the smallest incidence tree-depth of a matrix row-equivalent to $A$.

A \emph{circuit} of a rational matrix $A$
is a support-wise minimal integral vector contained in the kernel of $A$ such that all its entries are coprime;
the set of circuits of $A$ is denoted by $\C(A)$.
Note that a set $X$ of columns is a circuit in the matroid $M(A)$ if and only if
$\C(A)$ contains a vector with the support exactly equal to $X$.
We write $c_1(A)$ for the maximum $\ell_1$-norm of a circuit of $A$ and
$c_{\infty}(A)$ for the maximum $\ell_{\infty}$-norm of a circuit of $A$.
Note if $A$ and $A'$ are row-equivalent rational matrices, then $\C(A)=\C(A')$ and
so the parameters $c_1(\cdot)$ and $c_\infty(\cdot)$ are invariant under elementary row operations.
Following the notation from~\cite{EkbNV21},
we write $\dot{\kappa}_A$ for the least common multiple of the entries of the circuits of $A$.
Observe that there exists a function $f:\NN\to\NN$ such that $\dot{\kappa}_A\le f(c_{\infty}(A))$ for every matrix $A$.

If $\vex$ and $\vey$ are two $d$-dimensional vectors,
we write $\vex \sqsubseteq \vey$
if $|\vex_i|\le|\vey_i|$ for all $i\in [d]$ and
$\vex$ and $\vey$ are in the same orthant,
i.e., $\vex_i$ and $\vey_i$ have the same sign (or one or both are zero) for all $i\in [d]$.
The \emph{Graver basis} of a matrix $A$, denoted by $\G(A)$,
is the set of the $\sqsubseteq$-minimal non-zero elements of the integer kernel $\ker_{\ZZ}(A)$.
We use $g_1(A)$ and $g_{\infty}(A)$ for the Graver basis of $A$ analogously to the set of circuits,
i.e., $g_1(A)$ is the maximum $\ell_1$-norm of a vector in $\G(A)$ and
$g_{\infty}(A)$ is the maximum $\ell_{\infty}$-norm of a vector in $\G(A)$.
Again, the parameters $g_1(\cdot)$ and $g_\infty(\cdot)$ are invariant under elementary row operations as
the Graver bases of row-equivalent matrices are the same.
Note that every circuit of a matrix $A$ belongs to the Graver basis of $A$, i.e., $\C(A)\subseteq\G(A)$, and
so it holds that $c_1(A)\le g_1(A)$ and $c_{\infty}(A)\le g_{\infty}(A)$ for every matrix $A$.

The existence of efficient algorithms for integer programming with of constraint matrices $A$ with bounded primal and dual tree-depth
is closely linked to bounds on the norm of elements of the Graver basis of $A$.
In particular, Kouteck\'y, Levin and Onn~\cite{KouLO18} established the following.

\begin{theorem}
\label{thm:boundg}
There exist functions $f_P,f_D:\NN^2\to\NN$ such that the following holds for every rational matrix $A$:
$g_{\infty}(A)\le f_P(\td_P(A),\ec(A))$ and $g_{1}(A)\le f_D(\td_D(A),\ec(A))$.
\end{theorem}

\subsection{Matroid depth parameters}

We now define matroid depth parameters that will be of importance further.
We start with the notion of deletion-depth and contraction-depth,
which were introduced by DeVos, Kwon and Oum~\cite{DevKO20}.

The \emph{deletion-depth} of a matroid $M$, denoted by $\dd(M)$, is defined recursively as follows:
\begin{itemize}
\item If $M$ has a single element, then $\dd(M)=1$.
\item If $M$ is not connected, then $\dd(M)$ is the maximum deletion-depth of a component of $M$.
\item Otherwise, $\dd(M)=1+\min\limits_{e\in M}\dd(M\setminus e)$,
      i.e., $\dd(M)$ is is 1 plus the minimum deletion-depth of $M\setminus e$
      where the minimum is taken over all elements $e$ of $M$.
\end{itemize}
A sequence of deletions of elements witnessing that the deletion-depth of a matroid $M$ is $\dd(M)$
can be visualized by a rooted tree, which we call a \emph{deletion-tree}, defined as follows.
If $M$ has a single element,
then the deletion-tree of $M$ consists of a single vertex labeled with the single element of $M$.
If $M$ is not connected,
then the deletion-tree is obtained by identifying the roots of deletion-trees of the components of $M$.
Otherwise, there exists an element $e$ of the matroid $M$ such that $\dd(M)=\dd(M\setminus e)+1$ and
the deletion-tree of $M$ is obtained from the deletion-tree of $M\setminus e$
by adding a new vertex adjacent to the root of the deletion-tree of $M\setminus e$,
changing the root of the tree to the newly added vertex and labeling the edge incident with it with the element $e$.
See Figure~\ref{fig:dctree} for an example.
Observe that the height of the deletion-tree is equal to the deletion-depth of $M$.
In what follows, we consider deletion-trees that need not to be of optimal height,
i.e., its edges can be labeled by a sequence of elements that decomposes a matroid $M$ in a way described
in the definition of the deletion-depth but its height is larger than $\dd(M)$.
In this more general setting, the deletion-depth of a matroid $M$ is the smallest height of a deletion-tree of $M$.

\begin{figure}
\begin{center}
\epsfbox{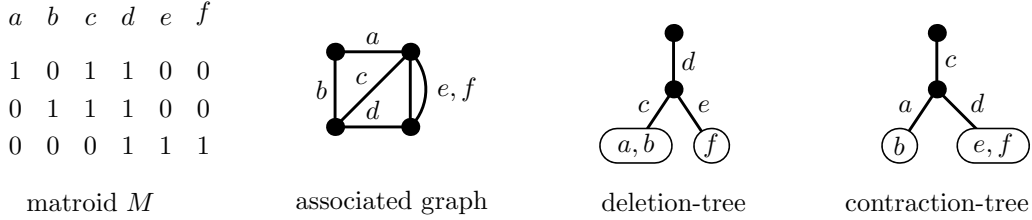}
\end{center}
\caption{A deletion-tree and a contraction-tree of the depicted binary matroid $M$,
         which is also the graphic matroid associated with the depicted graph.}
\label{fig:dctree}
\end{figure}

The \emph{contraction-depth} of a matroid $M$, denoted by $\cd(M)$, is defined recursively as follows:
\begin{itemize}
\item If $M$ has a single element, then $\cd(M)=1$.
\item If $M$ is not connected, then $\cd(M)$ is the maximum contraction-depth of a component of $M$.
\item Otherwise, $\cd(M)=1+\min\limits_{e\in M}\cd(M/e)$, 
      i.e., $\cd(M)$ is 1 plus the minimum contraction-depth of $M/e$ where the minimum is taken over all elements $e$ of $M$.
\end{itemize}      
It is not hard to show that $\dd(M)=\cd(M^*)$ and $\cd(M)=\dd(M^*)$ for every matroid $M$.
We define a \emph{contraction-tree} analogously to a deletion-tree;
the contraction-depth of a matroid $M$ is the smallest height of a contraction-tree of $M$ (an example
is given in Figure~\ref{fig:dctree}).

We next introduce the contraction-deletion-depth;
this parameter was studied under the name \emph{type} in~\cite{DinOO95},
however, we decided to adopt the name contraction-deletion-depth from~\cite{DevKO20},
which we find to better fit the context considered here.
The \emph{contraction-deletion-depth} of a matroid $M$, denoted by $\cdd(M)$, is defined recursively as follows:
\begin{itemize}
\item If $M$ has a single element, then $\cdd(M)=1$.
\item If $M$ is not connected, then $\cdd(M)$ is the maximum contraction-deletion-depth of a component of $M$.
\item Otherwise, $\cdd(M)=1+\min\limits_{e\in M}\min\{\cdd(M\setminus e),\cdd(M/e)\}$,
      i.e., $\cdd(M)$ is 1 plus the smaller among the minimum contraction-deletion-depth of the matroid $M\setminus e$ and
      the minimum contraction-deletion-depth of the matroid $M/e$ where both minima are taken over all elements $e$ of $M$.
\end{itemize}
Observe that it holds that $\cdd(M)=\cdd(M^*)$, $\cdd(M)\le\dd(M)$ and $\cdd(M)\le\cd(M)$ for every matroid $M$.

One of the key parameters in our setting is that of contraction$^*$-depth;
this parameter was introduced under the name branch-depth in~\cite{KarKLM17} and further studied in~\cite{ChaCKKP20}
but we decided to use a different name to avoid a possible confusion with the notion of branch-depth introduced in~\cite{DevKO20}.
We first introduce the parameter for general matroids, and
then present an equivalent definition for represented matroids,
which is more convenient to work in our setting.
The \emph{contraction$^*$-depth} of a matroid $M$, denoted by $\csd(M)$,
is the smallest depth of a rooted tree $T$ with exactly $r(M)$ edges with the following property:
there exists a function $f$ from the ground set of $M$ to the leaves of $T$ such that
for every subset $X$ of the ground set of $M$ the total number of edges contained in paths from the root to vertices of $X$
is at least $r(X)$.
An example is given in Figure~\ref{fig:csd}.

\begin{figure}
\begin{center}
\epsfbox{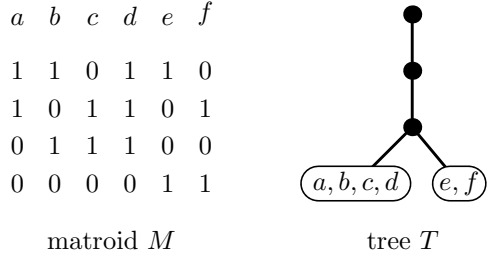}
\end{center}
\caption{A binary matroid $M$ of rank four given by its representation and
         a tree $T$ as in he definition of the contraction$^*$-depth.}
\label{fig:csd}
\end{figure}

There is an alternative definition of the parameter for represented matroids,
which also justifies the name that we use for the parameter.
The \emph{contraction$^*$-depth} of a represented matroid $M$ can be defined recursively as follows:
\begin{itemize}
\item If $M$ has rank zero, then $\csd(M)=0$.
\item If $M$ is not connected, then $\csd(M)$ is the maximum contraction$^*$-depth of a component of $M$.
\item Otherwise, $\csd(M)$ is 1 plus the minimum contraction$^*$-depth of a matroid obtained from the matroid $M$ by factoring along an arbitrary one-dimensional subspace.
\end{itemize}
As the contraction in the definition is allowed to be by an arbitrary one-dimensional subspace,
not only by a subspace generated by an element of $M$, it follows that $\csd(M)\le\cd(M)$.

The sequence of such contractions can be visualized by a \emph{contraction$^*$-tree} that
is defined in the same way as a contraction-tree except that
one-vertex trees are associated with matroids of rank zero (rather than matroids consisting of a single element) and
the edges are labeled by one-dimensional subspaces.
If each one-dimensional subspace that is the label of an edge of the tree
is generated by an element of the matroid $M$,
we say that the contraction$^*$-tree is \emph{principal} and
we view the edges of the tree as labeled by the corresponding elements of $M$.
Note that the minimum depth of a principal contraction$^*$-tree of a matroid $M$
is an upper bound on its contraction$^*$-depth,
however, in general, the contraction$^*$-depth of a matroid $M$ can be smaller than
the minimum depth of a principal contraction$^*$-tree of $M$.
We point out that the notions of principal contraction$^*$-trees and contraction-trees differ in a subtle but important way.
For example, if $M$ is a vector matroid of rank one containing a single element $e$,
its only contraction-tree consists of a root labeled by $e$
while its only contraction$^*$-tree has a root and a leaf adjacent to it and
the edge joining them is labeled with (the subspace generated by) $e$.
However, if $M$ is a vector matroid of rank one containing two parallel elements $e$ and $e'$,
any contraction-tree and any contraction$^*$-tree of $M$ has depth one.
Still, the minimum depth of a principal contraction$^*$-tree of $M$ is either $\cd(M)-1$ or $\cd(M)$.

Kardo\v s et al.~\cite{KarKLM17} established the connection between the contraction$^*$-depth and the existence of a long circuit,
which is described in Theorem~\ref{thm:circuit} below;
Theorem~\ref{thm:circuit} implies that $\cd(M)\le k^2+1$ where $k$ is the size of the largest circuit of $M$.

\begin{theorem}
\label{thm:circuit}
Let $M$ be a matroid and $k$ the size of its largest circuit.
It holds that $\log_2 k\le\csd(M)\le k^2$.
Moreover, there exists a polynomial-time algorithm that
for an input oracle-given matroid $M$ outputs a principal contraction$^*$-tree of depth at most $k^2$.
\end{theorem}

We next introduce the parameter of contraction$^*$-deletion-depth,
which we believe to have not been yet studied previously, but
which is particularly relevant in our context.
To avoid unnecessary technical issues, we introduce the parameter for represented matroids only.
The \emph{contraction$^*$-deletion-depth} of a represented matroid $M$, denoted by $\csdd(M)$, is defined recursively as follows:
\begin{itemize}
\item If $M$ has rank zero, then $\csdd(M)=0$;
\item if $M$ has a single non-loop element, then $\csdd(M)=1$.
\item If $M$ is not connected, then $\csdd(M)$ is the maximum contraction$^*$-deletion-depth of a component of $M$.
\item Otherwise, $\csdd(M)$ is 1 plus the smaller among the minimum contraction$^*$-deletion-depth of the matroid $M\setminus e$,
      where the minimum is taken over all elements of $M$, and
      the minimum contraction$^*$-deletion-depth of a matroid obtained from $M$ by factoring along an arbitrary one-dimensional subspace.
\end{itemize}
Observe that $\csdd(M)\le\cdd(M)$ and $\csdd(M)\le\csd(M)$ for every matroid $M$.

Finally, if $A$ is a matrix, the \emph{deletion-depth}, \emph{contraction-depth}, etc. of $A$
is the corresponding parameter of the vector matroid $M(A)$ formed by the columns of $A$, and
we write $\dd(A)$, $\cd(A)$, etc. for the deletion-depth, contraction-depth, etc. of the matrix $A$.
Observe that the deletion-depth, contraction-depth etc. of a matrix $A$ is invariant under elementary row operations as
elementary row operations preserve the matroid $M(A)$.

\section{Structural results}
\label{sec:struct}

In this section, we prove our structural results concerning optimal primal tree-depth and
optimal incidence tree-depth of a matrix.
We start with presenting an algorithm,
which uses a deletion-tree of the matroid associated with a given matrix
to construct a row-equivalent matrix with small primal tree-depth.

\begin{lemma}
\label{lm:tdP}
There exists a polynomial-time algorithm that
for an input matrix $A$ and a deletion-tree of $M(A)$ with height $d$
outputs a matrix $A'$ row-equivalent to $A$ such that $\td_P(A')\le d$.
\end{lemma}

\begin{proof}
We establish the existence of the algorithm by proving that $\td_P(A')\le d$ in a constructive (algorithmic) way.
Fix a matrix $A$ and a deletion-tree $T$ of $M(A)$ with height $d$.

Let $X$ be the set of non-zero columns that are labels of the vertices of $T$.
We show that the columns contained in $X$ form a basis of the column space of the matrix $A$.
As the matroid obtained from $M(A)$ by deleting the labels of the edges of $T$ has no component of size two or more,
the columns contained in $X$ are linearly independent.
Suppose that there exists a non-zero column $x$ that is not a linear combination of the columns contained in $X$, and
choose among all such columns the label of an edge $e$ as far from the root of $T$ as possible.
Since the element $x$ does not form its own component in the matroid obtained from $M$
by deleting the labels of all edges on the path from the root to $e$ (excluding $e$),
$x$ is a linear combination of the labels of the vertices and edges of the subtree of $T$ delimited by $e$.
This implies that either $x$ is a linear combination of the columns in $X$ or
there is a label of an edge of this subtree that is not a linear combination of the columns in $X$ contrary to the choice of $x$.
We conclude that $X$ is a basis of the column space of $A$.
In particular, unless $A$ is the zero matrix, the set $X$ is non-empty.

\begin{figure}
\begin{center}
\epsfbox{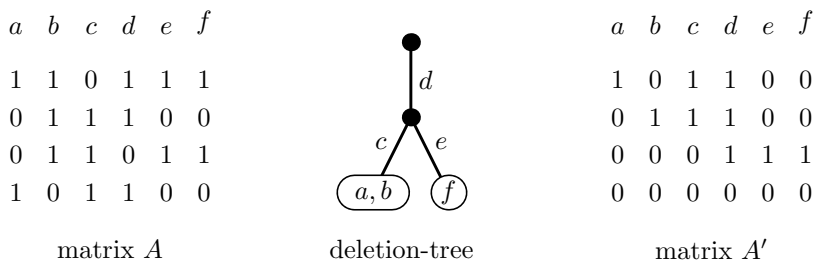}
\end{center}
\caption{A binary matrix $A$, a deletion-tree of the matroid $M(A)$ and the matrix $A'$
         as in the proof of Lemma~\ref{lm:tdP}.
	 Note that $X=\{a,b,f\}$.}
\label{fig:tdP}
\end{figure}

Let $A'$ be the matrix obtained from $A$ by elementary row operations such that
the submatrix of $A'$ induced by the columns of $X$
is the unit matrix and with some additional zero rows;
note that the set $X$ is determined by the input deletion-tree and so the matrix $A'$ depends on the deletion-tree.
See Figure~\ref{fig:tdP} for an example.
This finishes the construction of $A'$ and,
in particular, the description of the algorithm to produce the output matrix $A'$.
To complete the proof, it remains to show that the primal tree-depth of $A'$ is at most $d$,
i.e., the correctness of the algorithm.
This will be proven by induction on the number of columns of an input matrix $A$.

The base of the induction is the case when $A$ has a single column.
In this case, the primal tree-depth of $A'$ is one and
the tree $T$ is a single vertex labeled with the only column of $A$, and so its height is one.
We next present the induction step.
First observe that every label of a vertex of $T$ is either in $X$ or a loop in $M(A)$ (recall that 
only non-zero columns of $A$ are included to $X$), and
every label of an edge $e$ is a linear combination of labels of the vertices in the subtree delimited by $e$.

Suppose that the root of $T$ has a label and let $x$ be one of its labels;
note that $x$ is either a loop or a bridge in the matroid $M(A)$.
Let $B$ be the matrix obtained from $A$ by deleting the column $x$, and
let $T'$ be the deletion-tree of $M(B)$ obtained from $T$ by removing the label $x$ from the root.
Let $B'$ be the matrix produced by the algorithm described above for $B$ and $T'$.
If $x$ is a loop, then the matrix $A'$ is the matrix $B'$ extended by a zero column with possibly permuted rows, and
if $x$ is a bridge (and so $x\in X$),
then the matrix $A'$ is, possibly after permuting rows, the matrix $B'$ extended by a unit vector such that
its non-zero entry is the only non-zero entry in its row.
In either case, the vertex associated with the column $x$ is isolated in the primal graph of $A'$, and
it follows that $\td_P(A')=\td_P(B')\le d$ (the inequality holds by the induction hypothesis).
Hence, we can assume that the root of $T$ has no label.

We next analyze the case that the root of $T$ has a single child and no label.
Let $x$ be the label of the single edge incident with the root of $T$.
Let $B$ be the matrix obtained from $A$ by deleting the column $x$, and
let $T'$ be the deletion-tree of $M(B)$ obtained from $T$ by deleting the edge incident with the root and
rooting the tree at the remaining vertex of the deleted edge.
Let $B'$ be the matrix produced by the algorithm described above for $B$ and $T'$;
note that the primal tree-depth of $B'$ is at most $d-1$ by the induction hypothesis.
Since $B'$ is the submatrix of $A'$ formed by the columns different from $x$ (possibly after permuting rows),
the primal tree-depth of $A'$ is at most $\td_P(B')+1=d$.

The final case to analyze is the case
when the root of $T$ has $k\ge 2$ children (in addition to having no label).
Let $T_1,\ldots,T_k$ be the $k$ subtrees of $T$ delimited by the $k$ edges incident with the root of $T$,
let $Y_1,\ldots,Y_k$ be the labels of the vertices and edges of these subtrees, and
let $B_1,\ldots,B_k$ be the submatrices of $A$ formed by the columns contained in $Y_1,\ldots,Y_k$.
Observe that $T_i$ is a deletion-tree of the matroid $M(B_i)$ for $i=1,\ldots,k$, and
the matrix $B'_i$ produced by the algorithm described above for $B_i$ and $T_i$
is the submatrix of $A'$ formed by the columns contained in $Y_i$ (possibly after permuting rows).
Since the support of the columns contained in $Y_i$ contains only the unit entries of the columns of $A'$ contained in $X\cap Y_i$,
the primal graph of $A'$ contains no edge joining a column of $Y_i$ and a column of $Y_j$ for $i\not=j$.
It follows that the primal tree-depth of $A'$ is at most the maximum primal tree-depth of $B_i$,
which is at most $d$ by the induction hypothesis.
It follows that $\td_P(A')\le d$ as desired.
The proof of the correctness of the algorithm is now completed and so is the proof of the lemma.
\end{proof}

We are now ready to establish the link between the optimal primal tree-depth and the deletion-depth of the matroid associated with the matrix.

\begin{proof}[Proof of Theorem~\ref{thm:tdP}]
Fix a matrix $A$.
By Lemma~\ref{lm:tdP},
it holds that $\td_P^*(A)\le\dd(A)$
as there exists a deletion-tree of the matroid $M(A)$ with height $\dd(A)$.
So, we focus on proving that $\dd(A)\le\td_P^*(A)$.
We will show that every matrix $B$ satisfies that $\dd(B)\le\td_P(B)$;
this implies that $\dd(A)\le\td_P(A')$ for every matrix $A'$ row-equivalent to $A$ as $\dd(A)=\dd(A')$ and
so implies that $\dd(A)\le\td_P^*(A)$.

The proof that $\dd(B)\le\td_P(B)$ proceeds by induction on the number of columns.
If $B$ has a single column, then both $\dd(B)$ and $\td_P(B)$ are equal to one.
We next present the induction step.
We first consider the case when the matroid $M(B)$ is not connected.
Let $B_1,\ldots,B_k$ be the submatrices of $B$ formed by columns corresponding to the components of $M(B)$;
note that some of the submatrices may consist of a single zero column (if $M(B)$ has a loop).
The definition of the deletion-depth implies that $\dd(B)$ is the maximum among $\dd(B_1),\ldots,\dd(B_k)$.
On the other hand, the primal tree-depth of each of the matrices $B_i$
is at most the primal tree-depth of the matrix $B$ as
the primal graph of $B_i$ is a subgraph of the primal graph of $B$.
It follows that $\dd(B_i)\le\td_P(B)$, which implies that $\dd(B)\le\td_P(B)$.

We next assume that the matroid $M(B)$ is connected and
claim that the primal graph of $B$ must also be connected.
Suppose that the primal graph of $B$ is not connected,
i.e., there exists a partition of rows of $B$ into $R_1$ and $R_2$ and
a partition of the columns into $C_1$ and $C_2$, such that for each $i=1,2$,
the support of each column in $C_i$ is contained in $R_i$.
Therefore,
for any dependent set of columns of $B$,
either its subset formed by columns contained in $C_1$ is dependent, or
its subset formed by by columns contained in $C_2$ is dependent, or
both these subsets are dependent.
It follows that the support of every circuit of $M(B)$ is fully contained in either $C_1$ or $C_2$;
in particular, no there is no circuit of $M(B)$ containing a column from $C_1$ and a column from $C_2$.
This implies that $M(B)$ is not connected.
Hence, the primal graph of $B$ must be connected.

Since the primal graph of $B$ is connected,
there exists a column such that the matrix $B'$ obtained by deleting this column satisfies that $\td_P(B')=\td_P(B)-1$.
The induction assumption yields that $\dd(B')\le\td_P(B)-1$ and
the definition of the deletion-depth yields that
the deletion-depth of $M(B)$ is at most the deletion-depth of $M(B')$ increased by one.
This implies that $\dd(B)=\dd(M(B))\le\td_P(B)$ as desired.
\end{proof}

Before proceeding with our structural result concerning incidence tree-depth,
we use the structural results presented in Lemma~\ref{lm:tdP} and Theorem~\ref{thm:tdP}
to get a parameterized algorithm for computing an optimal primal tree-depth of a matrix over a finite field.

\begin{corollary}
\label{cor:tdP}
There exists a fixed parameter algorithm for the parameterization by a finite field $\FF$ and an integer $d$ that
for an input matrix $A$ over the field $\FF$,
\begin{itemize}
\item either outputs that $\td_P^*(A)>d$, or
\item computes a matrix $A'$ row-equivalent to $A$ with $\td_P(A')\le d$ and
      also outputs the associated deletion-tree of $M(A)$ with height $\td_P(A')$.
\end{itemize}
\end{corollary}

\begin{proof}
We first show that
the property that a matroid $M$ has deletion depth at most $d$ can be expressed in monadic second order logic.
Recall that monadic second order formulas for matroids
may contain quantification over elements and subsets of elements of a matroid, and
the predicate $\psi_I(\cdot)$ used to test whether a particular subset is independent
in addition to logic connectives and the equality $=$, the set inclusion $\in$ and the subset inclusion $\subseteq$.
In the formulas that we present,
small letters are used to denote elements of a matroid and capital letters subsets of the elements.
We next present a monadic second order formula $\psi_d(X)$ that
describes whether the deletion-depth of the restriction of the matroid $M$ to a subset $X$ of the elements of $M$ is at most $d$,
which would imply the statement.
The following formula $\psi_c(\cdot,\cdot)$ describes the existence of a circuit containing two distinct elements:
\[\psi_c(x,y)\equiv \left(x\not=y\right)\land\left(\exists X: x\in X\land y\in X\land\neg\psi_I(X)\land\forall z\in X:\psi_I(X\setminus z)\,\right).\]
The next formula $\psi_C(\cdot)$ describes whether a subset $X$ is a component of a matroid (recall that
the binary relation of two matroid elements being contained in a common circuit is transitive):
\[\psi_C(X)\equiv \left(\forall x,y\in X:x\not=y\Rightarrow\psi_c(x,y)\,\right)\land
                  \left(\forall x\in X, y\not\in X:\neg\psi_c(x,y)\,\right).\]
The sought formula $\psi_d(\cdot)$ is defined inductively as follows (we remark that $\psi_d(\emptyset)$ is true for all $d$):
\begin{align*}
\psi_1(X) & \equiv \forall x,y\in X: x\not=y\Rightarrow\neg\psi_c(x,y) &&\mbox{and}\\
\psi_d(X) & \equiv \forall X'\subseteq X: \psi_C(X')\Rightarrow\exists x\in X':\psi_{d-1}(X'\setminus\{x\}) &&\mbox{for $d\ge 2$.}
\end{align*}

Hlin\v en\'y~\cite{Hli03a,Hli06} proved that all monadic second order logic properties
can be tested in a fixed parameter way for matroids represented over a finite field $\FF$ with branch-width at most $d$
when parameterized by the property, the field $\FF$ and the branch-width $d$.
Since the branch-width of a matroid $M$ is at most its deletion-depth,
it follows that testing whether the deletion-depth of an input matroid represented over a finite field $\FF$ is at most $d$
is fixed parameter tractable when parameterized by the field $\FF$ and the integer $d$.
This establishes the existence of a fixed parameter algorithm deciding
whether $\td_P^*(A)=\dd(M(A))\le d$ (the equality follows from Theorem~\ref{thm:tdP}).
To obtain the algorithm claimed to exist in the statement of the corollary,
we need to extend the algorithm for testing
whether the deletion-depth of an input matroid $M$ represented over $\FF$ is at most $d$ to
an algorithm that also outputs a deletion-tree of $M$ with height at most $d$;
this would yield an algorithm for computing $A'$ by Lemma~\ref{lm:tdP}.

We now describe the extension of the algorithm from testing to constructing a deletion-decompo\-si\-tion tree as a recursive algorithm.
Let $d$ be the computed deletion-depth of an input matroid $M$.
The deletion-depth of an input matroid $M$ is one if and only if every element of $M$ is a loop or a bridge.
Since the latter is easy to algorithmically test,
if $d=1$, then the deletion-tree of height one consists of a single vertex labeled with all elements of $M$.
If $d\ge 2$ and the matroid $M$ is not connected,
we first identify its components, which can be done in polynomial time even in the oracle model,
then proceed recursively with each component of $M$ and
eventually merge the roots of all deletion-trees obtained recursively.

Finally, we discuss the case when $d\ge 2$ and the matroid $M$ is connected.
We loop over all elements $e$ of $M$ and test using the monadic second order checking algorithm of Hlin\v en\'y~\cite{Hli03a,Hli06}
whether $\dd(M\setminus e)\le d-1$, i.e., whether $\psi_{d-1}(M\setminus e)$ is true.
Such an element $e$ must exist (otherwise, the deletion-depth of $M$ cannot be $d$) and when $e$ is found,
we recursively apply the algorithm to $M\setminus e$
to obtain a deletion-tree $T$ of $M\setminus e$ with height $d-1$.
Note that there is a single recursive call as we invoke recursion for a single element $e$ of the matroid $M$.
The tree $T$ returned by the recursive call
is extended to a deletion-tree of $M$ by introducing a new vertex,
joining it by an edge to the root of $T$,
rerooting the tree to the new vertex, and
labeling the new edge with the element $e$.
This completes the description of the algorithm for constructing a deletion-tree of height at most $d$ if it exists.
Observe that the running time of the described procedure for constructing a deletion-tree
is bounded by the product of the number of elements of $M$ and
the running time of the test whether an input matroid is connected and
the test whether an input matroid satisfies $\psi_1(\cdot),\ldots,\psi_d(\cdot)$;
in particular, it is fixed parameter when parameterized by a finite field $\FF$ and an integer $d$.
\end{proof}

We conclude this section by establishing a link between the optimal incidence tree-depth and
the contraction$^*$-deletion-depth of the matroid associated with the matrix.

\begin{proof}[Proof of Theorem~\ref{thm:tdI}]
We prove the equality as two inequalities starting with the inequality $\csdd(A)\le\td_I^*(A)-1$.
To prove this inequality,
we show that $\csdd(A)\le\td_I(A)-1$ holds for every matrix $A$ with $m$ rows and $n$ columns by induction on $m+n$.
The base of the induction is formed by the cases when all entries of $A$ are zero, $n=1$ or $m=1$.
If all entries of $A$ are zero, then $\csdd(A)=0$ and $\td_I(A)=1$.
If $n=1$ and $A$ is non-zero,
then $M(A)$ has a single non-loop element and so $\csdd(A)=1$
while the incidence graph of $A$ is formed by a star and possibly some isolated vertices and so $\td_I(A)=2$.
Finally,
if $m=1$ and $A$ is non-zero,
then $M(A)$ has rank 1, so contracting any non-loop element of $M(A)$ yields a
matroid of rank zero and so $\csdd(A)=1$.
On the other hand,
the incidence graph of $A$ is formed by a star and possibly some isolated vertices and so $\td_I(A)=2$.

We now establish the induction step,
i.e., we assume that the matrix $A$ is non-zero, $m\ge 2$ and $n\ge 2$.
First suppose that the matroid $M(A)$ is not connected.
Let $X_1,\ldots,X_k$ be the components of $M(A)$ and
let $A_1,\ldots,A_k$ be the submatrices of $A$ formed by the columns $X_1,\ldots,X_k$, respectively.
The definition of the contraction$^*$-deletion-depth implies that
$\csdd(A)$ is the maximum of $\csdd(A_i)$.
The induction hypothesis yields that $\csdd(A_i)\le\td_I(A_i)-1$.
Since the incidence graph of $A_i$ is a subgraph of the incidence graph of $A$,
it follows that $\td_I(A_i)\le\td_I(A)$ and so $\csdd(A_i)\le\td_I(A)-1$.
We conclude that $\csdd(A)\le\td_I(A)-1$.

Next suppose that the matroid $M(A)$ is connected
but the incidence graph of $A$ is not connected.
As the columns associated with vertices contained in different components of the incidence graph of $A$ have disjoint supports,
such columns cannot be contained in the same component of the matroid $M(A)$.
Hence, if the incidence graph of $A$ is not connected despite the matroid $M(A)$ being connected,
then the incidence graph of $A$ consists of a single non-trivial component and
isolated vertices associated with zero rows of $A$.
Let $x$ be one such row and
let $A'$ be the matrix obtained from $A$ by deleting the row $x$.
Since the matroids $M(A)$ and $M(A')$ are the same,
it holds that $\csdd(A)=\csdd(A')$, and
since the incidence graph of $A$ is the incidence graph of $A'$ with an isolated vertex added,
it holds $\td_I(A)=\td_I(A')$.
Hence, the induction hypothesis yields that $\csdd(A')\le\td_I(A')-1$,
which implies that $\csdd(A)\le\td_I(A)-1$.

Finally, suppose that the matroid $M(A)$ is connected and
the incidence graph of $A$ is also connected.
The definition of the tree-depth implies that
there exists a vertex $w$ of the incidence graph whose deletion decreases the tree-depth of the incidence graph by one.
Let $A'$ be the matrix obtained from $A$ by deleting the row or the column associated with the vertex $w$ and
note that $\td_I(A')=\td_I(A)-1$.
If the vertex $w$ is associated with a column $x$,
the matroid $M(A')$ is the matroid obtained from $M(A)$ by deleting the element $x$.
If the vertex $w$ is associated with a row $x$,
the matroid $M(A')$ is the matroid obtained from $M(A)$ by contracting
by the subspace generated by the unit vector with the entry in the row $x$.
In either case,
the definition of the contraction$^*$-deletion-depth implies that $\csdd(A)\le\csdd(A')+1$.
The induction hypothesis applied to $A'$ yields that $\csdd(A')\le\td_I(A')-1$,
which yields that $\csdd(A)\le\td_I(A')=\td_I(A)-1$.

To complete the proof of the theorem,
it remains to show that
the inequality $\td_I^*(A)\le\csdd(A)+1$ holds for every matrix $A$.
The proof proceeds by induction on the number $n$ of columns of $A$.
If $n=1$ and the only column of $A$ is zero,
then the incidence graph of $A$ is formed by isolated vertices and so $\td_I(A)=1$
while $\csdd(A)=0$ since the rank of $M(A)$ is zero.
If $n=1$ and the only column of $A$ is not zero,
then the incidence graph of $A$ is formed by a star and possibly some isolated vertices and so $\td_I(A)=2$
while $\csdd(A)=1$.
In either case, it holds that $\td_I^*(A)\le\td_I(A)=\csdd(A)+1$.

We now establish the induction step.
First suppose that the matroid $M(A)$ is not connected.
Let $X_1,\ldots,X_k$ be the sets of columns forming the components of $M(A)$ and
let $A'$ be the matrix row-equivalent to $A$ such that
there exist sets of rows $Y_1,\ldots,Y_k$ such that
$|Y_i|=r_{M(A)}(X_i)$ for $i=1,\ldots,k$ and
the only columns with non-zero entries in the rows $Y_i$ are those of $X_i$ and
all rows not contained in $Y_1\cup\cdots\cup Y_k$ are zero (such a matrix $A'$ exists
since the matroid $M(A)$ is union of its restrictions to $X_1,\ldots,X_k$).
Let $A'_i$ be the submatrix of $A'$ formed by the rows of $Y_i$ and the columns of $X_i$.
Observe that all entries of the matrix $A$ not contained in one of the matrices $A'_1,\ldots,A'_k$ are zero.
By the induction hypothesis, for every $i=1,\ldots,k$,
there exists a matrix $A''_i$ row-equivalent to $A'_i$ such that $\td_I(A''_i)\le\csdd(M(A)\left[X_i\right])+1$,
in particular, $\td_I(A''_i)\le\csdd(A)+1$.
Let $A''$ be the matrix obtained from $A'$ by replacing $A'_i$ with $A''_i$ for $i=1,\ldots,k$.
Observe that $A''$ is row-equivalent to $A'$ and so to $A$.
Since the incidence graph of $A''$
is the union of the incidence graphs of $A''_i$, $i=1,\ldots,k$, and possibly some isolated vertices (which correspond to zero rows),
it follows that $\td_I(A'')\le\csdd(A)+1$.
Hence, it holds that $\td_I^*(A)\le\csdd(A)+1$.

To complete the proof,
we need to consider the case that the matroid $M(A)$ is connected.
The definition of the contraction$^*$-deletion-depth implies that
there exists an element $x$ of $M(A)$ such that $\csdd(M(A)\setminus x)=\csdd(M(A))-1=\csdd(A)-1$ or
there exists a one-dimensional subspace such that the contraction by this subspace
yields a matroid $M'$ such that $\csdd(M')=\csdd(M(A))-1=\csdd(A)-1$.
In the former case, let $A'$ be the matrix obtained from $A$ by deleting the column $x$.
By the induction hypothesis,
there exists a matrix $A''$ row-equivalent to $A'$ such that $\td_I(A'')\le\csdd(A')+1=\csdd(A)$, and
let $A'''$ be the matrix obtained from $A$ by the same elementary row operations that $A''$ is obtained from $A'$.
Observe that the incidence graph of $A''$
can be obtained from the incidence graph of $A'''$ by deleting the vertex associated with the column $x$.
Hence, $\td_I(A''')\le\td_I(A'')+1$.
Since $A'''$ is row-equivalent to $A$, it follows that
\[\td_I^*(A)\le\td_I(A''')\le\td_I(A'')+1\le\csdd(A)+1.\]

We now analyze the latter case, i.e., the case that
there exists a one-dimensional subspace such that the contraction by this subspace
yields a matroid $M'$ with $\csdd(M')=\csdd(A)-1$.
Let $A'$ be the matrix obtained from $A$ by elementary row operations such that the contracted subspace used to obtain $M'$
is generated by the unit vector with the non-zero entry being its first entry, and
let $B$ be the matrix obtained from $A'$ by deleting the first row.
By the induction hypothesis, there exists a matrix $B'$ row-equivalent to $B$ such that $\td_I(B')\le\csdd(A')+1=\csdd(A)$, and
let $A''$ be the matrix consisting of the first row of $A$ and the matrix $B'$.
Observe that $A''$ is row-equivalent to $A$.
Since the incidence graph of $B'$
can be obtained from the incidence graph of $A''$ by deleting the vertex associated with the first row,
it holds that $\td_I(A'')\le\td_I(B')+1$.
Hence, it follows that
\[\td_I^*(A)\le\td_I(A'')\le\td_I(B')+1\le\csdd(A)+1.\]
The proof of the theorem is now completed.
\end{proof}

\section{Primal tree-depth}
\label{sec:primal}

In this section, we present a parameterized algorithm for computing a row-equivalent matrix
with small primal tree-depth and bounded entry complexity if such a matrix exists.

\begin{proof}[Proof of Theorem~\ref{thm:alg-tdP}]
We first find a bound on $\dot{\kappa}_B$ when $B$ is an integer matrix with $\td_P(B)\le d$ and $\ec(B)\le e$ (recall
that $\dot{\kappa}_B$ is the least common multiple of the entries of the circuits of $B$).
Consider such a square invertible matrix $C$ with $\td_P(C)\le d$ and $\ec(C)\le e$.
By the result of Brand, Ordyniak and the second author~\cite{BraKO21},
the maximum denominator appearing over all entries of $C^{-1}$ can be bounded by a function of $d$ and $e$;
in particular, there exists $k_0\le (2^e)^{d!}{(d!)}^{d!/2}$ such that
every entry of $C^{-1}$ is $1/k$-integral for some $k\le k_0$.
Set $\kappa_0$ to be the least common multiple of the integers $1, \ldots, k_0$.
By \cite[Theorem 3.8]{EkbNV21},
$\dot{\kappa}_B$ is the smallest integer such that
every denominator of the inverse of every invertible square submatrix of $B$ divides $\dot{\kappa}_B$.
Since the primal tree-depth of any square submatrix of $B$ at most the primal tree-depth of $B$,
$\dot{\kappa}_B$ divides $\kappa_0$ for every matrix $B$ with $\td_P(B)\le d$ and $\ec(B)\le e$.
In particular, $\dot{\kappa}_B\le\kappa_0$ for every matrix $B$ with $\td_P(B)\le d$ and $\ec(B)\le e$.

We next describe the algorithm from the statement of the theorem.
Without loss of generality, we can assume that
the rank of the input matrix $A$ is equal to the number of its rows, in particular, $A$ is non-zero.
The algorithm starts with diagonalizing the square submatrix of the input matrix $A$
formed by an arbitrary basis of the column space,
i.e., performing elementary row operations so that the selected columns form the identity matrix.
The resulting matrix is denoted by $A_I$.
If the numerator or the denominator of any (non-zero) entry of $A_I$ does not divide $\kappa_0$,
the algorithm arrives at the first conclusion of the theorem as
every (non-zero) entry of $A_I$ is a fraction that can be obtained by dividing two entries of a circuit of $A$ (indeed,
consider the circuit of the matrix $A$ with support formed by a column $x$ and some of the columns of the chosen basis, and
observe that each entry in the column $x$
is equal to the $x$-entry of the circuit divided by one of its other entries).
Hence, we assume that both numerator and denominator of each (non-zero) entry of $A_I$ divides $\kappa_0$ in the rest.
The algorithm multiplies $A_I$ by $\kappa_0$,
which yields an integral matrix $A_0$ with entries between $-\kappa_0^2$ and $\kappa_0^2$.

Let $M_{\QQ}$ be the column matroid of $A_0$ when viewed as a matrix over rationals and
let $M_p$ be the column matroid of $A_0$ when viewed as a matrix over a $p$-element field $\FF_p$ for a prime $p>\kappa_0^2$;
note that such a prime $p$ can be found algorithmically as
the algorithm is parameterized by $d$ and $e$ and $\kappa_0$ depends on $d$ and $e$ only.
Note that the elements of both matroids $M_{\QQ}$ and $M_p$ are the columns of the matrix $A_0$,
i.e., we can assume that their ground sets are the same, and
the matroid $M_{\QQ}$ is the column matroid of $A$, which is a matrix over rationals.

We now establish the following claim:
\emph{if $A$ is row-equivalent to a matrix with primal tree-depth at most $d$ and entry complexity at most $e$,
then the matroids $M_{\QQ}$ and $M_p$ are the same}.
If a set $X$ of columns forms a circuit in $M_{\QQ}$,
then there exists a linear combination of the columns of $X$ that
has all coefficients integral and coprime,
i.e., not all are divisible by $p$, and that is equal to the zero vector (in fact, there exist such
coefficients that they all divide $\kappa_0$ by the definition of $\kappa_0$);
it follows that the set $X$ is also dependent in $M_p$.
If a set $X$ of columns is independent in $M_{\QQ}$,
let $B_I$ be an invertible square submatrix of $A_I$ formed by the columns $X$ and $|X|$ rows, and
let $B_0$ be the square submatrix of $A_0$ formed by the same rows and columns.
By~\cite[Lemma~3.3]{EkbNV21},
the matrix $B_I^{-1}$ is $1/\dot{\kappa}_A$-integral and
the absolute value of both numerator and denominator of each entry of $B_I^{-1}$ is at most $\dot{\kappa}_A$.
Note that $\dot{\kappa}_A$ divides $\kappa_0$ (here,
we use the definition of $\kappa_0$ and the assumption that
$A$ is row-equivalent to a matrix with primal tree-depth at most $d$ and entry complexity at most $e$) and
so the matrix $B_I^{-1}$ is $1/\kappa_0$-integral and
the absolute value of both numerator and denominator of each entry of $B_I^{-1}$ is at most $\kappa_0$.
Let $B'$ be the matrix obtained from $B_I^{-1}$ by multiplying each entry by $\kappa_0$;
note that $B'$ is an integer matrix with entries between $-\kappa_0^2$ and $\kappa_0^2$.
The definitions of the matrices $B_I$, $B_0$ and $B'$ yield that
the product of the matrix $B'$ when viewed as over $\FF_p$ and the matrix $B_0$
has numerically the same entries as the matrix $(\kappa_0 B_I^{-1})(\kappa_0 B_I)$,
which is a diagonal matrix with all diagonal entries equal to $\kappa_0^2$.
Hence, the matrix $B_0$ is full rank (over the field $\FF_p$) and
so the set $X$ is independent in $M_p$.

We now continue the description of the algorithm that is asserted to exist in the statement of the theorem.
As the next step, we apply the algorithm from Corollary~\ref{cor:tdP} to the matrix $A_0$ viewed as over the field $\FF_p$.
If the algorithm determines that the deletion-depth of $A_0$ is larger than $d$,
we arrive at the first conclusion of the theorem:
either the matroids $M_{\QQ}$ and $M_p$ are different (and so $A$ is not row-equivalent to a matrix with primal tree-depth at most $d$ and entry complexity at most $e$), or
the matroids $M_{\QQ}$ and $M_p$ are the same but $\dd(A)=\td_P^*(A)>d$.
If the algorithm determines that the deletion-depth of $A_0$ is at most $d$,
it also outputs a deletion-tree of $M_p$ with height at most $d$.
If the output deletion-tree is not valid for the matroid $M_{\QQ}$,
which can be verified in polynomial time,
the matroids $M_{\QQ}$ and $M_p$ are different and we again arrive at the first conclusion of the theorem.
If the output deletion-tree is also a deletion-tree of the matroid $M_{\QQ}$,
we use the algorithm from Lemma~\ref{lm:tdP}
to obtain a matrix $A'$ row-equivalent to $A$ such that
the primal tree-depth of $A'$ at most the height of the deletion-tree, i.e., $\td_P(A')\le d$.
As the matrix $A'$ contains a unit submatrix formed by $m$ rows and $m$ columns,
each (non-zero) entry of $A'$ is a fraction that can be obtained by dividing two entries of a circuit of $A$ (as argued earlier).
If the absolute value of the numerator or the denominator of any of these fractions exceeds $\kappa_0$,
then $c_\infty(A)>\kappa_0$ and we again arrive at the first conclusion of the theorem.
Otherwise, we output the matrix $A'$.
Note that the primal tree-depth of $A'$ is at most $d$ and
its entry complexity is at most $2\left\lceil\log_2\left(\kappa_0+1\right)\right\rceil$.
As $\kappa_0$ depends on $d$ and $e$ only,
the matrix $A'$ has the properties given in the second conclusion of the theorem.
\end{proof}

\section{Dual tree-depth, circuit complexity and Graver basis}
\label{sec:complex}

In this section, we link minimum dual tree-depth of a matrix to its circuit complexity, and
also present related algorithmic results.
We start with proving Theorem~\ref{thm:equiv2};
in fact, we show that a matrix row-equivalent to an input matrix $A$ such that
both its dual tree-depth and entry complexity bounded by a function of $c_1(A)$
can be found efficiently.
Note that the algorithm presented in the next theorem is not a fixed parameter algorithm,
i.e., its running time is polynomial in the size of an input matrix.

\begin{theorem}
\label{thm:alg-tdD-c1}
There exists a polynomial-time algorithm that for a given rational matrix $A$ with $\dim\ker A>0$
outputs a row-equivalent matrix $A'$ such that $\td_D(A')\le c_1(A)^2$ and $\ec(A')\le 2\left\lceil\log_2 (c_1(A)+1)\right\rceil$.
\end{theorem}

\begin{proof}
We start with the description of the algorithm from the statement of the theorem.
Let $A$ be the input matrix.
We apply the algorithm from Theorem~\ref{thm:circuit} to the matroid $M(A)$ given by the columns of the matrix $A$.
Let $T$ be the principal contraction$^*$-tree output by the algorithm and
let $X$ be the set of columns of $A$ that are the labels of the edges of $T$,
i.e., the elements of $M(A)$ used in the contractions.
Observe that the definition of the contraction$^*$-depth and the principal contraction$^*$-tree yields that
the labels of the edges of $T$ form a basis $X$ of $M(A)$ and
for every element $z$ of $M(A)$ there is a leaf $v$ of $T$ such that
$z$ is contained in the linear hull of the labels of the edges on the path from $v$ to the root of $T$.
Next perform row-operations on the matrix $A$ in a way that
the submatrix formed by the columns of $X$ is an identity matrix (with additional zero rows
if the rank of $A$ is smaller than the number of its rows);
let $A'$ be the resulting matrix;
see Figure~\ref{fig:tdD} for an example.
The algorithm outputs the matrix $A'$.
Note that the running time of the algorithm is indeed polynomial in the size of the input matrix $A$.

\begin{figure}
\begin{center}
\epsfbox{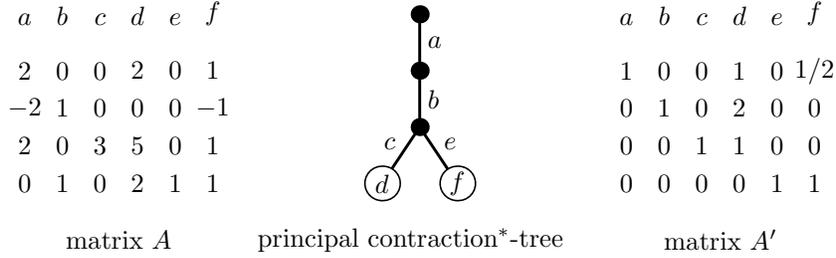}
\end{center}
\caption{A rational matrix $A$, a principal contraction$^*$-tree $T$ of the matroid $M(A)$ and the matrix $A'$
         as in the proof of Theorem~\ref{thm:alg-tdD-c1}.}
\label{fig:tdD}
\end{figure}

We next analyze the matrix $A'$ that is output by the algorithm.
Since $\dim\ker A>0$,
the matrix $A$ has at least one circuit.
Recall that for every circuit $C$ of the matroid $M(A)$,
there exists a circuit of $A$ whose support is exactly formed by the elements of $C$.
This implies that every circuit of $M(A)$ contains at most $c_1(A)$ elements and so
Theorem \ref{thm:circuit} implies that the depth of the principal contraction$^*$-tree $T$ is at most $c_1(A)^2$.

Let $F$ be a rooted forest obtained from the tree $T$ as follows.
For each edge $e$, the vertex of $e$ farther from the root is identified with the (unique) row of $A'$ that
is non-zero in the column that is the label of the edge $e$ (recall that the columns of $X$ form an identity matrix), and
then remove the root of $T$;
also add an isolated vertex for each zero row of $A'$.
In this way, we obtain a rooted forest $F$ with vertex set formed by the rows of $A'$.
Note that the height of $F$ is at most $c_1(A)^2$.
We will show that the dual graph of $A'$ is a subgraph of $\cl(F)$.
As no column of $X$ contributes any edges to the dual graph of $A'$,
it is enough to consider columns not contained in $X$.
Let $z$ be a column of $A'$ that is not contained in $X$ and
let $v$ be a leaf of $T$ such that
the column $z$ of $A$, which is an element of $M(A)$,
is contained in the linear hull of the labels of the edges on the path from $v$ to the root of $T$.
Hence, the column $z$ of $A'$ contains non-zero entries only in the rows with non-zero entries in the columns that
are labels of the edges on the path from $v$ to the root of $T$.
Consequently, all edges contained in the dual graph of $A'$ because of non-zero entries in the column $z$
are between vertices on the path from the vertex $v$ in $F$ to the root of the corresponding tree of $F$.
It follows that the dual graph of $A'$ is a subgraph of $\cl(F)$ and
so its tree-depth is at most the height of $F$, i.e., it is at most $c_1(A)^2$.
We conclude that $\td_D(A')\le c_1(A)^2$.

It remains to analyze the entry complexity of $A'$.
The entries of $A'$ in the columns of $X$ are zero or one.
Next consider a column $z$ of $A'$ that is not contained in $X$ and
consider a circuit $c$ of $A'$ (and so of $A$)
whose support contains $z$ and some elements of $X$ (such a circuit exists as the columns of $X$
form a basis of the column space of $A'$).
Observe that the entries in the column $z$ are equal to $-c_x/c_z$ (otherwise, $c$ would not be a circuit of $A'$).
We conclude that the entry complexity of $A'$ is at most $2\left\lceil\log_2 (c_1(A)+1)\right\rceil$.
\end{proof}

We are now ready to prove Theorem~\ref{thm:equiv1}.
Note that the condition $\dim\Ker A>0$ in the statement of the theorem 
is necessary as otherwise $A$ has no circuit and so $c_1(A)$ is not defined.

\begin{proof}[Proof of Theorem~\ref{thm:equiv1}]
Consider a rational matrix $A$ with $\dim\Ker A>0$.
Note that $c_1(A)\le g_1(A)$ as every circuit of $A$ is also an element of the Graver basis of $A$.
To prove the existence of the function $f_1$,
let $f_D$ be the function from Theorem~\ref{thm:boundg} and
note that Theorem~\ref{thm:alg-tdD-c1} implies that $g_1(A)\le f_D\left(c_1(A)^2,2\left\lceil\log_2 (c_1(A)+1)\right\rceil\right)$.
\end{proof}

We next combine the algorithms from Theorems~\ref{thm:orig-alg} and~\ref{thm:alg-tdD-c1}.

\begin{corollary}
\label{cor:alg-tdD-opt}
There exists a function $f:\NN\to\NN$ and a fixed parameter algorithm for the parameterization by $k$ that
for a given rational matrix $A$:
\begin{itemize}
\item either outputs that $c_1(A)>k$, or
\item outputs a matrix $A'$ that is row-equivalent to $A$,
      its dual tree-depth is $\td_D^*(A)$ and its entry complexity is at most $f(k)$.
\end{itemize}
\end{corollary}

\begin{proof}
If $\dim\Ker A=0$ for an input matrix $A$,
i.e., the rank of $A$ is equal to the number of the columns, which can be easily verified in polynomial time,
then $A$ is row-equivalent to the unit matrix possibly with some zero rows added,
i.e., to a matrix with dual tree-depth one and entry complexity one.
If $\dim\Ker A>0$,
we apply the algorithm from Theorem~\ref{thm:alg-tdD-c1} to get a matrix $A'$ row-equivalent to $A$ that
has properties given in the statement of Theorem~\ref{thm:alg-tdD-c1}.
If the dual tree-depth of $A'$ is larger than $k^2$ or the entry complexity of $A'$ is larger than $2\lceil\log_2 (k+1)\rceil$,
then $c_1(A)>k$ (by Theorem~\ref{thm:alg-tdD-c1}) and we arrive at the first conclusion.
Otherwise, we apply the algorithm from Theorem~\ref{thm:orig-alg} with parameters $d=k^2$ and $e=2\lceil\log_2 (k+1)\rceil$
to compute a matrix $A''$ row-equivalent to $A'$ and so to $A$ such that
the dual tree-depth of $A''$ is $\td_D^*(A)$ and
the entry complexity of $A''$ is bounded by a function of $k$ only.
\end{proof}

Finally, the previous corollary together with Theorem~\ref{thm:boundg} yields the parameterized algorithm
for testing whether an input matrix is row-equivalent to a matrix with small dual tree-depth and small entry complexity
as given in Theorem \ref{thm:alg-tdD}.

\begin{proof}[Proof of Theorem~\ref{thm:alg-tdD}]
Let $f_D$ be the function from the statement of Theorem~\ref{thm:boundg} and set $k=f_D(d,e)$;
note $f_D(d,e)\le 2^{2^{(d \log e)^{O(1)}}}$ by Eisenbrand et al.~\cite{EisHKKLO19}.
Apply the algorithm from Corollary~\ref{cor:alg-tdD-opt} with the parameter $k$ to an input matrix $A$.
If the algorithm reports that $c_1(A)>k$,
then $A$ is not row-equivalent to a matrix with dual tree-depth at most $d$ and entry complexity at most $e$.
If the algorithm outputs a matrix $A'$ and $\td_D(A')>d$, then $\td_D^*(A)>d$ and
so the matrix $A$ is not row-equivalent to a matrix with dual tree-depth at most $d$.
Otherwise, the dual tree-depth of $A'$ is at most $d$ and
its entry complexity is bounded by $f(k)=f(f_D(d,e))$
where $f(\cdot)$ is the function from Corollary~\ref{cor:alg-tdD-opt},
i.e., the entry complexity of $A'$ is bounded by a function of $d$ and $e$ only as required.
\end{proof}

\section{Computational hardness of depth parameters}
\label{sec:hardness}

In this section,
we complement our algorithmic results by establishing computational hardness of matroid depth parameters
that we have discussed in this paper.
The hardness results apply even when the input matroid is given by its representation over a fixed (finite or infinite) field.

We start with defining a matroid $M_{\FF}(G)$ derived from a graph $G$.
Fix a field $\FF$.
For a graph $G$, we define an $\FF$-represented matroid $M_{\FF}(G)$ as follows.
The matroid $M_{\FF}(G)$ contains $|V(G)|+|E(G)|$ elements,
each corresponding to a vertex or an edge of $G$.
We associate each element of $M_{\FF}(G)$ with a vector of $\FF^{V(G)}$.
An element of $M_{\FF}(G)$ corresponding to a vertex $w$ of $G$ is represented by $e_w$ and
an element of $M_{\FF}(G)$ corresponding to an edge $ww'$ of $G$ is represented by $e_w-e_{w'}$ or $e_{w'}-e_w$ (an arbitrary
of the two vectors can be chosen as the choice does not affect the matroid).

We next define a graph $G/A$ for a graph $G$ and a linear subspace $A$ of $\FF^{V(G)}$.
Let $W$ be the subset of vertices of $V(G)$ such that $e_w\in A$ for $w\in W$, and
let $F$ be the set of edges $ww'$ of $G$ such that
neither $w$ nor $w'$ is contained in $W$ and
$A$ contains a vector $e_w+\alpha e_{w'}$ for a non-zero element $\alpha\in\FF$.
The graph $G/A$ is obtained by deleting all vertices of $W$ and
then contracting a maximal acyclic subset of edges contained in $F$ (the remaining edges of $F$
become loops and so get removed);
note that we use an acyclic subset of edges for contraction so that the resulting graph is well-defined.

The next lemma relates
the number of components of the matroid $M_{\FF}(G)/A$ and the number of components of the graph $G/A$
for a graph $G$ and a linear subspace $A$ of $\FF^{V(G)}$.

\begin{lemma}
\label{lm:A-contract}
Let $G$ be a graph, $\FF$ a field and $A$ a linear subspace of $\FF^{V(G)}$.
The number of components of $M_{\FF}(G)/A$ is at most the number of components of the graph $G/A$.
\end{lemma}

\begin{proof}
Fix a graph $G$, a field $\FF$ and a linear subspace of $\FF^{V(G)}$.
Let $W$ and $F$ be the subsets of vertices and edges of $G$ as in the definition of $G/A$, respectively.
Let $A'$ be the linear subspace of $\FF^{V(G)}$ generated
by the vectors $e_w$, $w\in W$, and the vectors $e_w+\alpha e_{w'}\in A$ for $ww'\in F$;
clearly, $A'$ is a subspace of $A$.
Note that the sets $W$ and $F$ defined with respect to $A$ would be the same if defined with respect to $A'$, and
so the graphs $G/A$ and $G/A'$ are the same.
We next describe an $\FF$-representation of the matroid $M_{\FF}(G)/A'$ using vectors of $\FF^{V(G/A')}$.
The matroid $M_{\FF}(G)$ contains elements corresponding to vertices and to edges of $G$.
Consider a vertex $w$ of $V(G)$.
If $w\in W$, then the element corresponding to $w$ is a loop in $M_{\FF}(G)/A'$ and so represented by the zero vector.
If $w\not\in W$, then the element corresponding to $w$ is represented by the vector $e_{u}$
where $u$ is the vertex of $G/A'$ that the vertex $w$ was contracted to.
Next consider an edge $ww'$ of $G$.
\begin{itemize}
\item
    If both $w$ and $w'$ belong to $W$,
    then the element corresponding to $ww'$ is a loop in $M_{\FF}(G)/A'$ and so represented by the zero vector.
\item
    If exactly one of $w$ and $w'$ belong to $W$, say $w\in W$ and $w'\not\in W$,
    then the element corresponding to $ww'$ is represented by the vector $e_{u}$
    where $u$ is the vertex of $G/A'$ that the vertex $w'$ was contracted to.
\item
    If neither $w$ nor $w'$ belongs to $W$ and $e_w-e_{w'}\in A'$ (and so $ww'\in F$),
    then the element corresponding to $ww'$
    is a loop in $M_{\FF}(G)/A'$ and so represented by the zero vector.
\item
    If neither $w$ nor $w'$ belongs to $W$, $e_w-e_{w'}\not\in A'$ but $ww'\in F$,
    then the element corresponding to $ww'$ is represented by the vector $e_{u}$
    where $u$ is the vertex of $G/A'$ that the vertex $w$ (and so the vertex $w'$) was contracted to.
\item
    Finally, if neither $w$ nor $w'$ belongs to $W$ and $ww'\not\in F$,
    the element corresponding to $ww'$ is the vector $\alpha e_u+\alpha' e_{u'}$
    where $u$ is the vertex of $G/A'$ that the vertex $w$ was contracted to,
    $u'$ is the vertex of $G/A'$ that the vertex $w'$ was contracted to, and
    the vector $\alpha e_u+\alpha' e_{u'}$ corresponds to the vector $e_w-e_{w'}$ in $\FF^{V(G)}/A'$;
    note that $uu'$ is an edge of $G/A'$ and both coefficients $\alpha$ and $\alpha'$ are non-zero.
\end{itemize}
It is straightforward to verify that the just described representation
is indeed a representation of the matroid $M_{\FF}(G)/A'$;
note that each non-loop element of $M_{\FF}(G)/A'$ is associated to a vertex or an edge of $G/A'$ and
each vertex or an edge of $G/A'$ with at least one (but possibly more) non-loop element of $M_{\FF}(G)/A'$.

We now analyze the matroid $M_{\FF}(G)/A$.
The matroid $M_{\FF}(G)/A$ can be viewed as the matroid $(M_{\FF}(G)/A')/(A/A')$.
Observe that the definition of $A'$ implies that any non-loop element of $M_{\FF}(G)/A'$ is non-loop in $M_{\FF}(G)/A$.
Also observe that the space $A/A'$ viewed as a subspace of $\FF^{V(G/A')}$
does not contain a vector $e_w$ for $w\in V(G/A')$ or
a vector $e_w+\alpha e_{w'}$ for a non-zero $\alpha\in \FF$ such that $ww'$ is an edge of $E(G/A')$.
In particular, the support of every vector of $A/A'$ is at least two and
if the support has size two, then it does not correspond to an edge of $E(G/A')$.
It follows that if $ww'$ is an edge of $E(G/A')$,
$x$ is an element of $M_{\FF}(G)/A'$ associated with the vertex $w$,
$x'$ is an element of $M_{\FF}(G)/A'$ associated with the vertex $w'$, and
$x''$ is an element of $M_{\FF}(G)/A'$ associated with the edge $ww'$,
then the elements $x$, $x'$ and $x''$ form a circuit of $(M_{\FF}(G)/A')/(A/A')$.
Since the relation of being contained in a common circuit is transitive,
it follows that all elements of $M_{\FF}(G)/A'$ corresponding to the vertices and the edges of the same component of $G/A'$
are contained in the same component of $(M_{\FF}(G)/A')/(A/A')$.
In particular, the number of components of the matroid $(M_{\FF}(G)/A')/(A/A')$
is at most the number of components of the graph $G/A'$.
Since the graphs $G/A$ and $G/A'$ are the same and
the matroids $M_{\FF}(G)/A$ and $(M_{\FF}(G)/A')/(A/A')$ are also the same,
the lemma follows.
\end{proof}

We next link the existence of a balanced independent set in a bipartite graph
to the contraction$^*$-depth of a suitably defined matroid.
We remark that the idea of using a bipartite graph with cliques added between the vertices of its parts
was used in~\cite{pothen-td-hardness} to establish that computing tree-depth of a graph is NP-complete.

\begin{lemma}
\label{lm:graph-reduction}
Let $G$ be a bipartite graph with parts $X$ and $Y$, let $\FF$ be a field, and let $k$ be an integer.
Let $G'$ be the graph obtained from $G$ by adding all edges between the vertices of $X$ and between the vertices of $Y$.
The following three statements are row-equivalent.
\begin{itemize}
\item The graph $G$ has an independent set containing $k$ elements of $X$ and $k$ elements of $Y$.
\item The contraction$^*$-depth of $M_{\FF}(G')$ is at most $|X|+|Y|-k$.
\item The contraction-depth of the matroid $2M_{\FF}(G')$ is at most $|X|+|Y|-k+1$.
\end{itemize}
\end{lemma}

\begin{proof}
Fix a bipartite graph $G$ with parts $X$ and $Y$, a field $\FF$ and an integer $k$.
We first show that if $G$ has an independent set containing $k$ elements of $X$ and $k$ elements of $Y$,
then the contraction$^*$-depth of $M_{\FF}(G')$ is at most $|X|+|Y|-k$ and
the contraction-depth of $2M_{\FF}(G')$ is at most $|X|+|Y|-k+1$.
Let $W$ be such an independent set and
let $W'$ be the set containing all elements $e_w$ of $M_{\FF}(G')$ such that $w\not\in W$.
Note that $|W'|=|X|+|Y|-2k$.
The matroid $M_{\FF}(G')/W'$ is the matroid obtained from $M_{\FF}(G'[W])$ by adding
\begin{itemize}
\item a loop for every edge with both end vertices not contained in $W$, and 
\item an element represented by $e_w$ for every edge joining a vertex $w\in W$ to a vertex not contained in $W$.
\end{itemize}
In particular, the matroid $M_{\FF}(G')/W'$ has two non-trivial components, each of rank $k$, and
so the contraction$^*$-depth of $M_{\FF}(G')$ is at most $|W'| + k = |X| + |Y| - k$.
Similarly, the matroid $(2M_{\FF}(G'))/W'$ is the matroid obtained from $2M_{\FF}(G'[W])$ by adding
\begin{itemize}
\item a loop for every vertex not contained in $W$,
\item two loops for every edge with both end vertices not contained in $W$, and
\item two elements represented by $e_w$ for every edge joining a vertex $w\in W$ to a vertex not contained in $W$.
\end{itemize}
Since the matroid $(2M_{\FF}(G'))/W'$ has two non-trivial components, each of rank $k$,
its contraction-depth is at most $|W'|+k+1=|X|+|Y|-k+1$ (note that any rank $r$ matroid
has contraction-depth at most $r+1$).

We next argue that
if the contraction$^*$-depth of $M_{\FF}(G')$ is at most $|X|+|Y|-k$ or
the contraction-depth of $2M_{\FF}(G')$ is at most $|X|+|Y|-k+1$,
then there exists a subset $W$ of $V(G)=V(G')$ that is independent in $G$, $|W\cap X|\ge k$ and $|W\cap Y|\ge k$.
To do so,
we first show that there is no linear subspace $A$ such that $M_{\FF}(G')/A$ would have more than two components.
Consider a linear subspace $A$ of $\FF^{V(G')}$ such that the matroid $M_{\FF}(G')/A$ is not connected.
By Lemma~\ref{lm:A-contract}, the graph $G'/A$ is disconnected.
Since the graph $G'/A$ cannot have more than two components (one is formed by some of the vertices of $X$ and
another by some of the vertices of $Y$),
it follows that
the graph $G'/A$ has exactly two components and so the matroid $M_{\FF}(G')/A$ has exactly two components, too.

If the contraction$^*$-depth of $M_{\FF}(G')$ is at most $|X|+|Y|-k$,
there exists a linear subspace $A$ of $\FF^{V(G')}$ such that the matroid $M_{\FF}(G')/A$ is not connected and
the rank of each of its two components is at most $|X|+|Y|-k-\dim A$.
We will prove that the existence of such $A$ is also implied
by the assumption that the contraction-depth of $2M_{\FF}(G')$ is at most $|X|+|Y|-k+1$.
We next use that the matroid $(2M_{\FF}(G'))/F$ has at most two non-trivial components
for every subset $F$ of the elements of $2M_{\FF}(G')$.
If the contraction-depth of $2M_{\FF}(G')$ is at most $|X|+|Y|-k+1$,
then there exists a subset $F$ of the elements of $2M_{\FF}(G')$ such that
the matroid $(2M_{\FF}(G'))/F$ is not connected and
the rank of each of its two components is at most $|X|+|Y|-k-\rank F$ 
(as the contraction-depth of each of its two components is the rank of the component increased by one
because each element is parallel to at least one other element).
It follows that there exists a linear subspace $A$ of $\FF^{V(G')}$,
which is the hull of the vectors representing the elements of the set $F$ as above, such that
the matroid $M_{\FF}(G')/A$ is not connected and
the rank of each of its two components is at most $|X|+|Y|-k-\dim A$.
We conclude that
if the contraction$^*$-depth of $M_{\FF}(G')$ is at most $|X|+|Y|-k$ or
if the contraction-depth of $2M_{\FF}(G')$ is at most $|X|+|Y|-k+1$,
then there exists a linear subspace $A$ of $\FF^{V(G')}$ such that the matroid $M_{\FF}(G')/A$ is not connected and
the rank of each of its two components is at most $|X|+|Y|-k-\dim A$.

It remains to show that
\emph{the existence of a subspace $A$ of $\FF^{V(G')}$ such that
the matroid $M_{\FF}(G')/A$ is not connected and
the rank of each of its two components is at most $|X|+|Y|-k-\dim A$
implies that there exists an independent set containing $k$ elements of $X$ and $k$ elements of $Y$}.
Fix such a subspace $A$.
Let $W$ be the set of vertices $w$ such that $e_w$ is contained in $A$ and
let $A_W$ be the subspace of $A$ generated by the vectors $e_w$, $w\in W$.
Since $G'/A$ is not connected, the graph $G'\setminus W$ is also not connected (recall that $G'/A$
is obtained by removing the vertices of $W$ and then contracting some edges).
By Lemma \ref{lm:A-contract} the matroid $M_{\FF}(G')/A_W$ is also not connected.
Since the space $A_W$ is a subspace of $A$,
the rank of each component of $M_{\FF}(G')/A_W$
is larger by at most $\dim A-\dim A_W$ compared to the corresponding component of $M_{\FF}(G')/A$.
Hence, the rank of each of the two components of $M_{\FF}(G')/A_W$ is at most $|X|+|Y|-k-\dim A_W=|X|+|Y|-k-|W|$.
It follows that each component of the graph $G'/A_W=G'\setminus W$ contains at most $|X|+|Y|-k-|W|$ vertices.
Since the sum of the sizes of the two components of $G'\setminus W$ is $|X|+|Y|-|W|$,
each component of $G'\setminus W$ has at least $k$ vertices.
In addition, the vertex set of each component of $G'\setminus W$ is either a subset of $X$ or a subset of $Y$,
which implies that there is no edge joining a vertex of $X\setminus W$ and a vertex of $Y\setminus W$ and
both sets $X\setminus W$ and $Y\setminus W$ have at least $k$ vertices.
Hence, the graph $G$ has an independent set containing $k$ elements of $X$ and $k$ elements of $Y$ (such an independent
set is a subset of $V(G)\setminus W$).
\end{proof}

We are now ready to state our hardness result.

\begin{theorem}
\label{thm:NPc}
For every field $\FF$,
each of the following five decision problems, whose input is an $\FF$-represented matroid $M$ and an integer $d$,
is NP-complete:
\begin{itemize}
\item Is the contraction-depth of $M$ at most $d$?
\item Is the contraction$^*$-depth of $M$ at most $d$?
\item Is the contraction-deletion-depth of $M$ at most $d$?
\item Is the contraction$^*$-deletion-depth of $M$ at most $d$?
\item Is the deletion-depth of $M$ at most $d$?
\end{itemize}
\end{theorem}

\begin{proof}
It is NP-complete to decide for a bipartite graph $G$ with parts $X$ and $Y$ and an integer $k$ whether
there exist $k$-element subsets $X'\subseteq X$ and $Y'\subseteq Y$ such that $X'\cup Y'$ is independent~\cite{pothen-td-hardness}.
For an input bipartite graph $G$,
let $G'$ be the graph obtained from $G$ by adding all edges between the vertices of $X$ and between the vertices of $Y$.
We claim that
the existence of such subsets $X'$ and $Y'$
is equivalent to each of the following four statements:
\begin{itemize}
\item The matroid $2M_{\FF}(G')$ has contraction-depth at most $|X|+|Y|-k+1$.
\item The matroid $M_{\FF}(G')$ has contraction$^*$-depth at most $|X|+|Y|-k$.
\item The matroid $(|V(G')|+1)M_{\FF}(G')$ has contraction-deletion-depth at most $|X|+|Y|-k+1$.
\item The matroid $(|V(G')|+1)M_{\FF}(G')$ has contraction$^*$-deletion-depth at most $|X|+|Y|-k$.
\end{itemize}
The equivalences to the first and second statements follow directly from Lemma~\ref{lm:graph-reduction}.
Since the rank of the matroid $(|V(G')|+1)M_{\FF}(G')$ is $|G'|$,
its contraction-deletion-depth is at most $|G'|+1$ and
its contraction$^*$-deletion-depth is at most $|G'|$.
As each element of the matroid $(|V(G')|+1)M_{\FF}(G')$ is parallel to (at least) $|V(G')|$ elements of the matroid,
it follows that the contraction-deletion-depth of $M_{\FF}(G')$ is the same as its contraction-depth and
its contraction$^*$-deletion-depth is the same as its contraction$^*$-depth.
Lemma~\ref{lm:graph-reduction} now implies the equivalence of the third and fourth statements.
As the matroids $2M_{\FF}(G')$, $M_{\FF}(G')$ and $(|V(G')|+1)M_{\FF}(G')$ can be easily constructed from the input graph $G$ in time polynomial in $|V(G)|$,
the NP-completeness of the first four problems listed in the statement of the theorem follows.

For an $\FF$-represented matroid $M$,
it is easy to construct an $\FF$-represented matroid $M^*$ that is dual to $M$
in time polynomial in the number of the elements of $M$~\cite[Chapter 2]{Oxl11}.
Since the contraction-depth of $M$ is equal to the deletion-depth of $M^*$,
it follows that the fifth problem listed in the statement of the theorem is also NP-complete.
\end{proof}

\section{Concluding remarks}

We would like to conclude with addressing three natural questions related to the work presented in this paper.

In Section~\ref{sec:complex},
we have given a structural characterization of matrices $A$ with $g_1(A)$ bounded by showing
that $g_1(A)$ is bounded if and only if
$A$ is row-equivalent to a matrix with small dual tree-depth and small entry complexity.
Unfortunately,
a similar (if and only if) characterization of matrices $A$ with $g_{\infty}(A)$ bounded does not seem to be in our reach.
\begin{problem}
Find a structural characterization of matrices $A$ with $g_{\infty}(A)$ bounded.
\end{problem}
In view of Theorem~\ref{thm:boundg}, it may be tempting to think that
such a characterization can involve matrices with bounded incidence tree-depth
as if a matrix $A$ has bounded primal tree-depth or it has bounded dual tree-depth, then $g_{\infty}(A)$ is bounded.
However, the following matrix $A$ has incidence tree-depth equal to $4$ and
yet $g_{\infty}(A)$ grows with the number $t$ of its columns;
in particular, the vector $(t-1,1,1,\ldots,1)$ is an element of its Graver basis as
it can be readily verified.
We remark that a similar matrix was used by Eiben et al.~\cite{EibGKOPW20} in their NP-completeness argument.
\[\left(
  \begin{matrix}
  1 & -1 & -1 & -1 & \cdots & -1 & -1 \\
  0 &  1 & -1 &  0 & \cdots &  0 &  0 \\
  0 &  1 &  0 & -1 & \cdots &  0 &  0 \\
  \vdots & \vdots & \vdots & & \ddots & & \vdots \\
  0 &  1 &  0 &  0 & \cdots & -1 &  0 \\
  0 &  1 &  0 &  0 & \cdots &  0 & -1
  \end{matrix}
  \right)\]

In Section~\ref{sec:struct}, we have given structural characterizations of matrices
that are row-equivalent to a matrix with small primal tree-depth or small incidence tree-depth,
which complements the characterization of matrices row-equivalent to a matrix with small dual tree-depth from~\cite{ChaCKKP19,ChaCKKP20}.
We have also presented fixed parameter algorithms (Theorems~\ref{thm:alg-tdP} and \ref{thm:alg-tdD})
for finding such a row-equivalent matrix with bounded entry complexity if one exists;
both of these algorithms are based on fixed parameter algorithms for finding deletion-depth decompositions and
contraction$^*$-depth decompositions of matroids over finite fields,
which are presented in Corollary~\ref{cor:tdP} in the case of deletion-depth and
in~\cite{ChaCKKP19,ChaCKKP20} in the case of contraction$^*$-depth.
We believe that similar techniques would lead
to a fixed parameter algorithm for contraction$^*$-deletion-depth decompositions of
matroids represented over a finite field (note that contraction$^*$-deletion-depth does not have an obvious description in monadic second order logic and
so the algorithmic results of Hlin\v en\'y~\cite{Hli03a,Hli06} do not readily apply in this setting).
However,
it is unclear whether such an algorithm would yield a fixed parameter algorithm for rational matrices as
we do not have structural results on the circuits of rational matrices with small incidence tree-depth,
which would reduce the case of rational matrices to those over finite fields.

Another natural question is whether
the upper bound on the depth of the principal contraction$^*$-tree given in Theorem~\ref{thm:circuit},
which is quadratic in the length of the longest circuit of a represented matroid, can be improved.
However, this turns out to be impossible as we now argue.
Since the minimum depth of a principal contraction$^*$-tree of a matroid $M$ differs from $\cd(M)$,
i.e. the minimum height of a contraction-tree of $M$, by at most one,
it is enough to construct a sequence of matroids $M_n$ such that
\begin{itemize}
\item the length of the longest circuit of \(M_n\) is is at most \(\calO(n)\), and
\item the contraction-depth of \(M_{n}\) is at least \(\Omega(n^2)\).
\end{itemize}
Hence, the quadratic dependence of the minimum depth in Theorem~\ref{thm:circuit} is optimal up to a constant factor.
Still, it can be the case that the bound on the contraction$^*$-depth can be improved.

The matroids \(M_{n}\) are the graphic matroids of graphs \(G_n\),
which are constructed inductively.
To facilitate the induction we will require slightly stronger properties.
Each of the graphs \(G_n\) contains two distinguished vertices,
denoted by \(r_n\) and \(b_n\), and the following holds:

\begin{enumerate}
    \item The length of any path in \(G_n\) between the vertices \(r_n\) and \(b_n\) is between $n$ and \(2n\).\label{con:path}
    \item The length of any circuit in \(M_n\) is at most \(4n\).\label{con:cycle}
    \item The contraction-depth of \(M_n\) is at least \(n\choose 2\).\label{con:cd}
\end{enumerate}

If $n=1$, we set $G_1$ to be the two-vertex graph formed by two parallel edges, and
\(r_1\) and \(b_1\) are chosen as the two vertices of $G_1$.
Note that the graph $G_1$ and the matroid $M_1=M(G_1)$ has the properties \eqref{con:path}, \eqref{con:cycle}, and \eqref{con:cd}.
To obtain \(G_n\),
we start with a cycle of length $2n$ and choose any two vertices at distance $n$ to be \(r_n\) and \(b_n\).
This cycle containing the vertices \(r_n\) and \(b_n\) will be referred to as the \emph{root cycle}.
We then add \(n\) copies of \(G_{n - 1}\),
connect the vertex \(r_{n - 1}\) in each copy to the vertex \(r_n\), and
connect the vertex \(b_{n - 1}\) in each copy to \(b_{n}\).
The construction is illustrated in Figure~\ref{fig:construction}.

\begin{figure}[ht]
    \centering
    \includegraphics[scale=0.4]{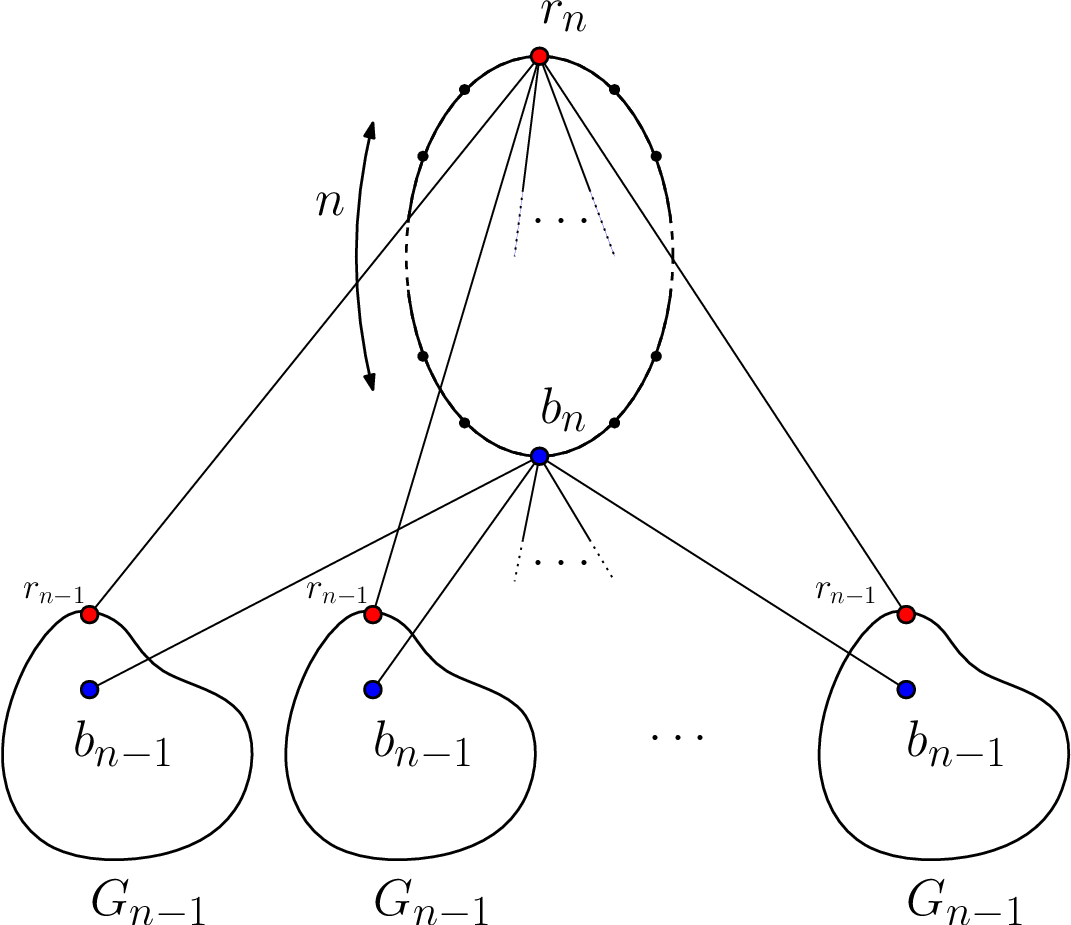}
    \caption{The construction of the graph \(G_{n}\).
             The vertices \(r_n\) and \(r_{n - 1}\) are drawn red
	     while \(b_n\) and \(b_{n - 1}\) are drawn blue.}
    \label{fig:construction}
\end{figure}

Assuming that the matroid $M_{n-1}$ and the graph $G_{n-1}$
have the properties \eqref{con:path}, \eqref{con:cycle}, and \eqref{con:cd}, and
we will show that the matroid $M_n$ and the graph $G_n$ also have these properties.

We start with showing that $G_n$ has the property~\eqref{con:path}.
Indeed, a path from \(r_n\) to \(b_n\) is either contained in the root cycle or
consists of a path between \(r_{n - 1}\) and \(b_{n - 1}\) in one of the copies of \(G_{n - 1}\),
whose length is between $n-1$ and $2(n-1)$, together with two edges joining
\(r_{n - 1}\) to \(r_{n}\) and \(b_{n - 1}\) to \(b_{n}\).
In either of the cases, the length of the path is between $n$ and $2n$ as required.

Having established the property~\eqref{con:path},
we prove the property~\eqref{con:cycle}.
Any circuit of the matroid \(M_{n}\) corresponds to a cycle in the graph $G_n$,
thus we can simply investigate the lengths of cycles in $G_n$.
First observe that a cycle of $G_n$ contains either both vertices \(r_n\) and \(b_n\) or neither of them.
Every cycle containing \(r_n\) and \(b_n\) consists of two paths between \(r_n\) and \(b_n\) and so its length is at most $4n$, and
every cycle containing neither \(r_n\) nor \(b_n\) is contained entirely within a copy of $G_{n-1}$ and so its length is at most $4(n-1)\leq 4n$.

Finally, we argue that contraction-depth of \(M_n\) is at least \({n} \choose 2\).
Recall that contracting an element of \(M_n\) corresponds to contracting the associated edge in the graph \(G_n\), and
components of a graphic matroid correspond to blocks, i.e., maximal $2$-edge-connected components, of an associated graph.
Since the length of any path between \(r_n\) to \(b_n\) is at least $n$,
until at least $n$ edge contractions are performed in graph \(G_n\),
the vertices \(r_n\) and \(b_n\) are distinct and are contained in the same block.
Hence,
after $n-1$ edge contractions followed by deleting all blocks not containing the vertices \(r_n\) and \(b_n\) (if such blocks appear),
the graph still contains an intact copy $G_{n-1}$.
It follows that $\cd(M_n)\ge (n-1)+\cd(M_{n-1})$,
which implies that $\cd(M_n)\ge (n-1)+{n-1\choose 2}={n\choose 2}$.

\section*{Acknowledgements}

All five authors would like to thank the Schloss Dagstuhl-–-Leibniz Center for Informatics for hospitality
during the workshop ``Sparsity in Algorithms, Combinatorics and Logic'' in September 2021
where the work leading to the results contained in this paper was started.
The authors are also indebted to the two anonymous reviewers for their detailed comments on the manuscript,
which have helped to improve the presentation significantly and made the manuscript more accessible to a wider audience.

\bibliographystyle{bibstyle}
\bibliography{ptdepth}

\end{document}